\documentclass[11pt,a4paper,reqno]{amsart}%
\usepackage{amsthm,amsmath,amsfonts,amssymb,amsxtra,appendix,bookmark,dsfont,bm,color}
\usepackage[shortlabels]{enumitem}

\theoremstyle{plain}
\newtheorem{theorem}{Theorem}[section]
\newtheorem{lemma}[theorem]{Lemma}

\theoremstyle{definition}

\newcounter{step} 

\usepackage{etoolbox}
\AtBeginEnvironment{proof}{\setcounter{step}{0}}

\DeclareMathOperator{\supp}{supp}

\def\bq{\begin{eqnarray}}
\def\eq{\end{eqnarray}}
\def\bqq{\begin{align*}}
\def\eqq{\end{align*}}

\def\nn{\nonumber}

\def\eps{\varepsilon}

\newcommand{\norm}[1]{\left\lVert #1 \right\rVert}

\newcommand\1{{\ensuremath {\mathds 1} }}

\def\cH{\mathcal{H}}

\def\R {\mathbb{R}}

\def\cE {\mathcal{E}}

\def\R {\mathbb{R}}

\def\d{{\, \rm d}}

\title[The ionization conjecture in TFDW theory]{The ionization conjecture in Thomas-Fermi-Dirac-von Weizs\"acker theory}

\author[R.L. Frank]{Rupert L. Frank}
\address{R. L. Frank, Mathematics 253-37, Caltech, Pasadena, CA 91125, USA} 
\email{rlfrank@caltech.edu}

\author[P.T. Nam]{Phan Th\`anh Nam}
\address{P.T. Nam, Institute of Science and Technology Austria, Am Campus 1, 3400 Klosterneuburg, Austria} 
\email{pnam@ist.ac.at}

\author[H. Van Den Bosch]{Hanne Van Den Bosch}
\address{H. Van Den Bosch, Instituto de F\'{i}sica, Pontificia Universidad Cat\'olica de Chile, Av. Vicu\~na Mackenna 4860, Santiago, Chile} 
\email{hannevdbosch@fis.puc.cl}

\begin{document}

\begin{abstract}
We prove that in Thomas-Fermi-Dirac-von Weizs\"acker theory, a nucleus of charge $Z>0$ can bind at most $Z+C$ electrons, where $C$ is a universal constant. This result is obtained through a comparison with Thomas-Fermi theory which, as a by-product, gives bounds on the screened nuclear 
potential and the radius of the minimizer. A key ingredient of the proof is a novel technique to control the particles in the exterior region, which also applies to the liquid drop model with a nuclear background potential. 
\end{abstract}

\date{March 27, 2017}

\makeatletter{\renewcommand*{\@makefnmark}{}
\footnotetext{\copyright\, 2017 by the authors. This paper may be reproduced, in its entirety, for non-commercial purposes.}\makeatother}

\maketitle

\setcounter{tocdepth}{1}
\tableofcontents
\addcontentsline{toc}{section}{Contents}

\section{Introduction}

It is well-known from experiments that a neutral atom can bind at most two extra electrons. However, justifying this fact rigorously from the first principles of quantum mechanics is difficult. This problem has been studied in many-body Schr\"odinger theory by many authors \cite{Sigal-82,Ruskai-82,Lieb-84,LieSigSimThi-88,FefSec-90,SecSigSol-90,Nam-12}. 
From these works (in particular, \cite{Lieb-84}, \cite{Nam-12} and \cite{FefSec-90,SecSigSol-90}), it is known that a nucleus of charge $Z$ can bind at most 
$$ \min\{ 2Z+1, 1.22~Z+3Z^{1/3}, Z+CZ^{5/7}+C\}$$
electrons, where $C$ is a universal constant. Establishing the bound $Z+C$ remains open and this is often referred to as the {\em ionization conjecture}, see \cite[Problem 9]{Simon-00} or \cite[Chapter 12]{LieSei-10}. 

While the full Schr\"odinger theory is very precise, it is too complicated for practical computations. Therefore, both qualitative and quantitative properties of  atoms are often studied using approximate theories. One of the most popular methods used in computational  physics and chemistry is density functional theory, where the properties of the many-body system are encoded in the electron density instead of the complex wavefunction. 

The oldest density functional theory is Thomas-Fermi (TF) theory \cite{Thomas-27,Fermi-27}, which goes back to the early days of quantum mechanics. The  TF functional is the semiclassical approximation of the many-body energy and it captures the leading order behavior of the many-body ground state energy in the large $Z$ limit \cite{LieSim-77b}. However, it has some qualitative defects, most notably the absence of negative ions (or more generally, Teller's no-binding theorem for molecules \cite{Teller-62}). 

The leading order correction to TF theory can be obtained by adding von Weizs\"acker's gradient term \cite{Weizsacker-35} to the energy functional. This term comes from the kinetic energy of the particles very close to the nucleus. It was proved in \cite{Lieb-81b} that Thomas-Fermi-von Weizs\"acker theory (with the appropriate constant in front of the gradient term) reproduces Scott's correction \cite{Scott-52} to the ground state energy (see \cite{Hu90,SW87,SW89} for the derivation of Scott's correction from the Schr\"odinger theory).  In Thomas-Fermi-von Weizs\"acker theory, negative ions exist \cite{BenBreLie-81} and the ionization conjecture  was proved by Benguria and Lieb \cite{BenLie-84} (see also \cite{Solovej-90}).

The second order correction to TF theory can be obtained by adding Dirac's term \cite{Dirac-30} to the energy functional. This correction comes both from the exchange energy and the semiclassical approximation. The resulting Thomas-Fermi-Dirac-von Weizs\"acker (TFDW) theory (with the appropriate constant in front of Dirac's term) is expected to reproduce the Dirac-Schwinger correction \cite{Schwinger-81} to the ground state energy (see \cite{FS90} for the derivation of the Dirac-Schwinger correction from Schr\"odinger theory). Thus, the accuracy of TFDW theory is comparable to that of Hartree-Fock theory \cite{Bach-92} in the large $Z$ regime, but the former is conceptually simpler because it only relies on electron densities rather than density matrices.

To be precise, we will consider the TFDW variational problem 
\bq \label{eq:variational-problem}
E_Z^{\rm TFDW}(N)=\inf \left\{ \cE_Z^{\rm TFDW}(\rho)\,:\, \rho\ge 0, \sqrt{\rho} \in H^1(\R^3), \int_{\R^3} \rho(x) \d x = N  \right\}
\eq
where
\begin{align*}
\cE_Z^{\rm TFDW}(\rho) &= c^{\rm TF} \int_{\R^3} \rho(x)^{5/3} \d x -  \int_{\R^3} \frac{Z \rho(x)}{|x|} \d x + \frac{1}{2} \iint_{\R^3\times \R^3} \frac{\rho(x)\rho(y)}{|x-y|} \d x \d y \\
& \qquad \qquad  + c^{\rm W} \int_{\R^3} |\nabla \sqrt {\rho(x)}|^2 \d x - c^{\rm D} \int_{\R^3} \rho(x)^{4/3} \d x.
\end{align*}
Of course, the positive constants $c^{\rm TF}$, $c^{\rm W}$ and $c^{\rm D}$ have to be chosen appropriately (see \cite{Lieb-81b,Schwinger-81}) to make TFDW theory a good approximation to Schr\"odinger theory. However, the specific values of these constants are not important for our analysis in this paper. Both of the nuclear charge $Z$ and the number of electrons $N$ are not necessarily integers.

In 1987, Lions \cite{Lions-87} proved that \eqref{eq:variational-problem} has a minimizer if $N\le Z$. The existence result was extended by Le Bris \cite{Bris-93} to all $N\le Z+\eps$ for some $\eps>0$, namely negative ions exist. On the other hand, the nonexistence for  $N$ large remains mostly open. In fact, the special case $Z=0$ is already subtle and it has been solved recently by Lu and Otto \cite{LuOtt-14}. This nonexistence result was extended by two of us  \cite{NamBos-16} to the case when $Z>0$ is very small (even in the molecular case, which we do not consider here).

Our main result in this paper is

\begin{theorem}[Ionization bound] \label{thm:ionization} There exists a constant $C>0$ such that for all $Z>0$, if $E_Z^{\rm TFDW}(N)$ in \eqref{eq:variational-problem} has a minimizer, then $N \le Z+C$. 
\end{theorem}

The main difficulty in TFDW theory is that the ``problem at infinity" has nontrivial bound states, and this makes it very challenging to control the particles escaping to infinity. In particular, the powerful argument of integrating the Euler-Lagrange equation against the moment $|x|$, which was used successfully in Thomas-Fermi-von Weizs\"acker theory \cite{Lieb-81b} as well as full Schr\"odinger theory \cite{Lieb-84}, is not applicable in TFDW theory (because $-\int |x|\rho(x)^{4/3} \d x$ can be very negative). On the other hand, the argument in \cite{LuOtt-14,NamBos-16} does not rely on the moment estimate but it only works when almost of all electrons escape to infinity, which requires that $Z$ is very small.

Thus to prove Theorem \ref{thm:ionization}, we need a novel method to control the particles far from the nucleus. We will use ideas in a recent work of R. Killip and two of us \cite{FraKilNam-16}, where the nonexistence in the liquid drop model was proved by dividing $\R^3$ by half-planes and taking the average. In Section \ref{sec:exterior-L1}, we will derive an upper bound on the number of electrons in the exterior region $|x|\ge r$. In particular, this  exterior $L^1$-estimate implies that $N\le 2Z+ C(Z^{2/3}+1)$. 
 
To prove $N\le Z+C$, we will employ the fact that the particles in the exterior region effectively feel the attraction of the nucleus screeened by the electrons in the interior region $|x|\le r$. We estimate the screened potential by comparing with TF theory, following Solovej's proof of the ionization conjecture in Hartree-Fock theory \cite{Solovej-03}. Our main technical tool is the following 

\begin{theorem}[Screened potential estimate] \label{thm:screened-intro} Let $\rho_0$ be a TFDW minimizer with some $N\ge Z\ge 1$. Let $\rho^{\rm TF}$ be the TF minimizer with $N=Z$ (see Theorem \ref{thm:TF-Z}). For every $r>0$, define the screened nuclear potentials
\begin{align*}
\Phi_r(x) = \frac{Z}{|x|} - \int_{|y|\le r} \frac{\rho_0(y)}{|x-y|} \d y, \qquad \Phi_r^{\rm TF}(x) = \frac{Z}{|x|} - \int_{|y|\le r} \frac{\rho^{\rm TF}(y)}{|x-y|} \d y.
\end{align*}
Then there are universal constants $C>0$, $\eps>0$ such that 
$$
\left|\Phi_{|x|}(x) -\Phi_{|x|}^{\rm TF}(x)\right| \le C (|x|^{-4+\eps}+1)
$$
for all $N \ge Z\ge 1$ and $|x|>0$.
\end{theorem}

The significance of the power $|x|^{-4+\eps}$ is that $\Phi_{|x|}^{\rm TF}(x) \sim |x|^{-4}$ for $|x|$ small (see Section \ref{sec:TF}). The bound in Theorem \ref{thm:screened-intro} for $|x|\le Z^{-1/3}$ follows easily from an energy comparison at the leading order. However, in order to extend this bound to all $|x|>0$, we need to use a delicate bootstrap argument which goes back to Solovej  \cite{Solovej-03}. 

By Newton's theorem \cite[Theorem 9.7]{LieLos-01}, we can write
\begin{align*} 
\int_{|y|<r} (\rho^{\rm TF}(y)-\rho_0(y)) \d y = r \int_{\mathbb{S}^2} (\Phi_r(r \nu)-\Phi_r^{\rm TF}(r\nu)) \frac{\d \nu}{4\pi}.
\end{align*}
Therefore, Theorem \ref{thm:screened-intro} allows us to control the number of electrons in the interior region $|x|<r$. Combining this with the exterior bound mentioned above, we conclude the ionization bound $N\le Z+C$ easily. Moreover, we can also deduce that the atomic {\em radius} in  TFDW theory is very close to that in TF theory. Similarly as in  \cite[Theorem 1.5]{Solovej-03}, we have the following asymptotic estimate for the radii of ``infinite atoms".

\begin{theorem}[Radius estimate] \label{thm:radius} Let $\rho_0$ be a TFDW minimizer with some $N\ge Z$. For $\kappa>0$, we define the radius $R(N,Z,\kappa)$ as the largest number such that
$$
\int_{|x|\ge R(N,Z,\kappa)} \rho_0(x) \d x = \kappa.
$$
Then there are universal constants $C>0$, $\eps>0$ such that 
$$
\limsup_{N\ge Z\to \infty}\left| R(N,Z,\kappa) - B^{\rm TF} \kappa^{-1/3}\right| \le C \kappa^{-1/3-\eps}
$$
for all $\kappa \ge C$, where $B^{\rm TF}=  5 c^{\rm TF} (4/(3\pi^2))^{1/3}$.  
\end{theorem} 

Our results can be extended partially to the case of molecules in TFDW theory. In particular, by adapting the exterior $L^1$-estimate, we  can show that the number of electrons in every molecule is bounded by a finite constant which depends only on the nuclear positions and charges (this result was proved in \cite{NamBos-16} under the extra assumption on the smallness of nuclear charges). Finding the asymptotic behavior of the maximum number of electrons when the nuclear charges become large in the case of molecules is an open problem. We hope to be able to come back to this issue in the future.   

We conclude this introduction with a related theorem in a different model, namely the liquid drop model with a nuclear background potential, which was recently proposed by \cite{LuOtt-15}. In contrast to the usual liquid drop model \cite[Eq. (4.1)]{ChoPel-10} (see also \cite{LuOtt-14,FraLie-15,FraKilNam-16} and the references therein), now the atom (not the nucleus) is assumed to have constant density. The kinetic energy of the electrons is modeled by a surface tension term in the energy functional. The variational problem is 
\bq \label{eq:liquid-drop-variational}
E_Z(N)=\inf \left\{ \cE(\Omega)\,:\, \Omega \subset \R^3 \,\text{measurable}, \, |\Omega|=N  \right\}
\eq
where 
$$
\cE_Z(\Omega) = |\partial \Omega| -Z\int_{\Omega}\frac{\d x}{|x|} + \frac{1}{2}\iint_{\Omega \times \Omega} \frac{\d x \d y}{|x-y|}.
$$
Here $|\partial \Omega|$ is the surface area of $\Omega$ if the boundary of $\Omega$ is smooth, and it is the perimeter in the sense of De Giorgi if the boundary is not smooth. Again, the parameters $N$ and $Z$ are not necessarily integers. 

We will prove  

\begin{theorem}[Nonexistence in liquid drop model] \label{thm:drop} If $E_Z(N)$ has a minimizer, then 
$N\le \min \{2Z+8, Z+ 8+ C Z^{1/3}\}$.
\end{theorem}
The first bound is a direct generalization of the case $Z=0$ in \cite{FraKilNam-16}. It is reminiscent of Lieb's bound $2Z+1$ on the number of electrons of atoms. The second bound improves the estimate $N \le Z+ C(Z^{2/3}+1)$ of Lu and Otto in \cite{LuOtt-15}.

While the physical significance of this model is not clear to us, it serves as a useful toy model for the more complicated TFDW problem. In particular, the proof of the exterior $L^1$-estimate is similar in both models. It is somewhat cleaner in the liquid drop case and therefore we present this first. Despite this similarity, we were not able to generalize the proof of the ionization conjecture to this model and we leave it as an open question whether the exponent in $Z^{1/3}$ can be improved.

\subsection*{Organization} The paper is organized as follows. We prove Theorem \ref{thm:drop} in Section \ref{sec:drop}. In the rest of the paper we concentrate on TFDW theory. In Section \ref{sec:exterior-L1}, we derive the exterior estimate for the number of electrons in the region $|x|>r$. As a corollary, we obtain a bound of the type $N \le C Z+C$, see Lemma~\ref{lem:2Z}.  In Section \ref{sec:TF}, we revisit TF theory. In Section \ref{sec:split}, we split the exterior region from the interior region in terms of energy contributions, 
for both TF and TFDW theories. With these preliminaries, in Section \ref{sec:bootstrap}, we prove the bound in Theorem \ref{thm:screened-intro} for $|x|\le O(1)$, using a bootstrap argument. Finally, in Section \ref{sec:proof-main-result} we conclude the proofs of Theorems \ref{thm:ionization}, \ref{thm:screened-intro} and \ref{thm:radius}.

\subsection*{Notations} We always denote by $C \ge 1$ a universal constant (whose value may change from line to line). We will use the short-hand notation
$$\mathfrak{D}(f)=\frac{1}{2} \iint_{\R^3\times \R^3} \frac{\overline{f(x)}f(y)}{|x-y|} \d x \d y.$$

\subsection*{Acknowledgements} We thank the referee for helpful suggestions which improved the presentation of the paper. Partial support by U.S. National Science Foundation DMS-1363432 (R.L.F.), Austrian Science Fund (FWF) Project Nr. P 27533-N27 (P.T.N.), CONICYT (Chile) through CONICYT--PCHA/Doctorado Nacional/2014 and Iniciativa Cient\'ifica Milenio (Chile) through Millenium Nucleus RC--120002 ``F\'isica Matem\'atica'' (H.V.D.B.) is acknowledged.

\section{Liquid drop model} \label{sec:drop}

In this section, we prove Theorem \ref{thm:drop}. We will denote by  $\chi_\Omega$ the characteristic function of $\Omega$ and by $\mathcal{H}^2$ the two-dimensional Hausdorff measure.
Moreover, for every $R>0$, we denote 
$$\Omega_{\ge R}=\Omega \cap \{|x| \ge R\},\quad \Omega_{\le R}=\Omega \cap \{|x| \le R\},\quad \Omega_{= R}=\Omega \cap \{|x| = R\}. $$

We have the following exterior $L^1$ estimate for the liquid drop model, which is based on ideas in \cite{FraKilNam-16}.

\begin{lemma}\label{lem:drop-binding} If $E_Z(N)$ has a minimizer $\Omega$, then for every $R \ge 0$ we have
\begin{align*}
&\frac{1}{2}|\Omega_{\ge R}|^2 +  \iint_{|x| \ge R \ge |y|} \frac{|x|}{|x-y|} \chi_\Omega(x) \chi_\Omega(y) \d x \d y \\
&\qquad \qquad \qquad\qquad \qquad \le (Z+4) |\Omega_{\ge R}| + 2 R \mathcal{H}^2 (\Omega_{=R}).
\end{align*}
\end{lemma}
\begin{proof} Let $\theta(x) = x \1(|x|\ge R)$. For every $\nu\in \mathbb{S}^2$ and $\ell>0$, we define
$$
H := \{ x\in\R^3 \,|\, \nu\cdot \theta(x) > \ell \} \,.
$$
By the minimality of $\Omega$, we obtain the binding inequality
\bq \label{eq:key-binding-Omega-0}
\cE_{Z}(\Omega) \le  \cE_{\rm Z}(H^c \cap \Omega) + \cE_{0} (H \cap \Omega).
\eq
Note that
$$
\partial H = \partial H^c = \{ |x|\ge R, \nu \cdot x =\ell\} \cup \{ |x|=R, \nu \cdot x \ge \ell \},
$$
and hence for almost every $\ell>0$,
\begin{align*}
&|\partial ( \Omega \cap H)| + |\partial (\Omega \cap H^c)| - |\partial \Omega| \le 2 \mathcal{H}^2(\Omega \cap \partial H)  \le \\
& \le 2 \mathcal{H}^2\Big(\Omega_{\ge R} \cap \{ \nu \cdot x = \ell \}\Big) + 2 \mathcal{H}^2\Big(\Omega_{= R} \cap \{ \nu \cdot x \ge \ell \}\Big).
\end{align*}
Here the first inequality is obtained similarly to the Lemma in \cite[p. 1034]{FraKilNam-16} (this holds for a.e. $\ell>0$), and the second inequality is simply the subadditivity of $\cH^2$. 
Thus the binding inequality \eqref{eq:key-binding-Omega-0} implies that
\begin{align*}
\iint_{\nu \cdot \theta(x) < \ell < \nu \cdot \theta(y)} &\frac{\chi_\Omega(x) \chi_\Omega(y)}{|x-y|} \d x \d y -  Z \int_{\nu\cdot \theta(x)>\ell} \frac{\chi_\Omega(x)}{|x|} \d x   \\
& \le 2 \mathcal{H}^2\Big(\Omega_{\ge R} \cap \{ \nu \cdot x = \ell \}\Big) + 2 \mathcal{H}^2\Big(\Omega_{= R} \cap \{ \nu \cdot x \ge \ell \}\Big)
\end{align*}
for a.e. $\ell>0$. Integrating over $\ell \in (0,\infty)$ we obtain 
\begin{align*}
&\int_0^\infty \left(\iint_{\nu \cdot \theta(x) < \ell < \nu \cdot \theta(y)} \frac{\chi_\Omega(x) \chi_\Omega(y)}{|x-y|} \d x \d y \right) \d \ell - Z \int \frac{[\nu \cdot x]_+}{|x|} \chi_\Omega(x) \d x\\
& \qquad\qquad  \le   2|\Omega_{\ge R} \cap \{ \nu \cdot x > 0 \}| + 2  \mathcal{H}^2(\Omega_{= R})[\nu \cdot x]_+
\end{align*}
where we have denoted $[a]_\pm =\max(\pm a, 0)$. Changing $\nu \mapsto -\nu$ and interchanging the role of $x$ and $y$ in the repulsion term gives us
\begin{align*}
&\int_0^\infty \left(\iint_{-\nu \cdot \theta(y) < \ell < -\nu \cdot \theta(x)} \frac{\chi_\Omega(x) \chi_\Omega(y)}{|x-y|} \d x \d y \right) \d \ell - Z \int \frac{[\nu \cdot x]_-}{|x|} \chi_\Omega(x) \d x\\
& \qquad\qquad  \le   2|\Omega_{\ge R} \cap \{ \nu \cdot x < 0 \}| + 2  \mathcal{H}^2(\Omega_{= R})[\nu \cdot x]_-.
\end{align*}
Summing the latter two inequalities and using
\begin{align*}
\int_0^\infty &\left[ \1\Big(\nu \cdot \theta(x) < \ell < \nu \cdot \theta (y) \Big) + \1\Big(-\nu \cdot \theta(y) < \ell < \nu \cdot \theta (x) \Big) \right] \d \ell \\
&\qquad\qquad \qquad  = \Big[\nu \cdot (\theta(y)-\theta(x)) \Big] _+
\end{align*}
we find that
\begin{align*}
&\iint \frac{[\nu \cdot (\theta(y)-\theta(x))]_+}{|x-y|} \chi_\Omega(x) \chi_\Omega(y) \d x \d y - Z \int \frac{|\nu \cdot x|}{|x|} \chi_\Omega(x) \d x  \\
&\le 2|\Omega_{\ge R}| + 2  \mathcal{H}^2 (\Omega_{= R})|\nu \cdot x|.
\end{align*}
Finally, we average over $\nu \in \mathbb{S}^2$ and use  
\bq \label{eq:average-nu}
\int_{\mathbb{S}^2} [\nu\cdot z]_+\,\frac{d\nu}{4\pi} =  \frac{|z|}{4},\quad \int_{\mathbb{S}^2} |\nu\cdot z| \,\frac{d\nu}{4\pi} =\frac{|z|}{2},  \quad \forall z\in \mathbb{R}^3.
\eq
This gives 
$$
\iint  \frac{|\theta(x)-\theta(y)|}{|x-y|} \chi_\Omega(x) \chi_\Omega(y) \d x \d y \le  (2Z+8) |\Omega_{\ge R}| + 4R  \mathcal{H}^2(\Omega_{= R}),
$$
which is equivalent to the desired inequality.
\end{proof}

From Lemma \ref{lem:drop-binding}, if we choose $R\to 0$, then we obtain immediately 
\bq \label{eq:drop-first-2Z+8}
N=|\Omega| \le 2Z+8.
\eq
This is the first bound in Theorem \ref{thm:drop}. To prove the second bound, we will show that $\Omega_{\le R}$ is close to a ball, which allows us to estimate the second term on the left side of the bound in Lemma \ref{lem:drop-binding}.  We are inspired by ideas in \cite{FefSec-90,SecSigSol-90}, where the asymptotic neutrality of atoms was proved by comparing the density of the many-body ground state with the Thomas-Fermi minimizer. 

In the following, we will denote by $R_Z=(3Z/(4\pi))^{1/3}$ the radius of a ball of volume $Z$ and $\chi_Z$ the characteristic function of the ball $B(0,R_Z)$.

\begin{lemma} \label{lem:chi-chiZ}Assume that $E_Z(N)$ has a minimizer $\Omega$ with $N= Z+Q$, $Z\ge 1$, $Q\ge 1$. Then for all $f \in H^1( \R^3)$ we have
$$
\left| \int f(x) (\chi_\Omega(x)-\chi_Z(x)) \d x \right| \le C\|\nabla f\|_{L^2} Q^{1/2}.
$$
Here $C$ is a universal constant (independent of $N,Z$ and $f$). 
\end{lemma}

\begin{proof} Our key estimate is 
\bq \label{eq:Coulomb-norm} 
\mathfrak{D}(\chi_\Omega-\chi_Z) \le CQ.
\eq
This bound can be found in \cite[Eq. (16)]{LuOtt-15}. Since its proof is simple, let us sketch it here. 
By the minimality of $\Omega$, 
\bq \label{eq:drop-binding-easy}
\cE_Z(\Omega) \le \cE_Z(\chi_Z) + E_0(Q).
\eq
Note that $E_0(Q)\le CQ$ because of the subadditivity of $E_0$. On the other hand, $|\partial \Omega|\ge |\partial B(0,R_Z)|$ by the isoperimetric inequality and the fact that $|\Omega|\geq |B(0,R_Z)|$. Moreover, 
\begin{align*}  
&\quad \mathfrak{D}(\chi_\Omega ) - \mathfrak{D}(\chi_Z) - \mathfrak{D}(\chi_\Omega-\chi_Z ) - Z \int \frac{\chi_\Omega(x)-\chi_Z(x)}{|x|} \d x \\
&= \iint \frac{(\chi_\Omega(x) -\chi_Z(x))\chi_Z(y)}{|x-y|} \d x \d y -\iint \frac{(\chi_\Omega(x)-\chi_Z(x))\chi_Z(y)}{|x|} \d x \d y  \nn \\
& = \iint (\chi_\Omega(x) -\chi_Z(x) ) \chi_Z(y) \left( \frac{1}{\max(|x|,|y|)} - \frac{1}{|x|}\right) \d x \d y \nn\\
& = \iint_{|x| \le |y| \le R_Z}  (\chi_\Omega(x) -1 )  \left( \frac{1}{|y|} - \frac{1}{|x|}\right) \d x \d y \ge 0.
\end{align*}
Here the second equality follows from Newton's theorem \cite[Theorem 9.7]{LieLos-01}, and the inequality is obvious since $\chi_\Omega-1 \le 0$. Thus \eqref{eq:drop-binding-easy} implies \eqref{eq:Coulomb-norm}.   

Next, we use a simple inequality whose relevance in a related problem has been noted by Fefferman and Seco \cite{FefSec-90},
\begin{align} \label{eq:Fef-Sec}
\left| \int f g \right| \le \left( \int |k|^2 |\widehat f(k) |^2 \d k \right)^{1/2} \left( \int \frac{|\widehat g(k) |^2}{|k|^2} \d k \right)^{1/2} =  \frac1{\sqrt{2\pi}} \|\nabla f\|_{L^2} \sqrt{\mathfrak{D}(g)}.
\end{align}
Using \eqref{eq:Fef-Sec} with $g=\chi_\Omega-\chi_Z$ and \eqref{eq:Coulomb-norm}, we obtain the desired estimate. 
\end{proof}

Now we are able to provide
\begin{proof}[Proof of Theorem \ref{thm:drop}] Assume that $E_Z(N)$ has a minimizer $\Omega$. We have already proved $N\le 2Z+8$ in \eqref{eq:drop-first-2Z+8}. Now we show that $N\le Z+8+CZ^{1/3}$. It suffices to consider the case $Z\ge 1$ and $Q=N-Z \ge 1$. 

We start with the key estimate in Lemma \ref{lem:drop-binding}:
\begin{align} \label{eq:key-estimate-in-proof}
&\frac{1}{2}|\Omega_{\ge R}|^2 +  \iint_{|x| \ge R \ge |y|} \frac{|x|}{|x-y|} \chi_\Omega(x) \chi_\Omega(y) \d x \d y \nn\\
&\qquad\qquad\qquad \le (Z+4) |\Omega_{\ge R}| + 2 R \mathcal{H}^2 (\Omega_{=R}), \quad \forall R>0.
\end{align}
Let $f_0:\R^3\to \R$ be a smooth, radially symmetric function such that
\bq \label{eq:def-f0}
0\le f_0 \le 1, \quad f_0(x) = 1 \,\,\, \text{if}\,\, |x| \le 1, \quad f_0(x) =0 \,\,\,\text{if}\,\, |x| \ge 2. 
\eq
We have
\begin{align}  \label{eq:drop-binding-mainterm}
\int_{|y| \le R} \frac{\chi_\Omega(y)}{|x-y|} \d y &\ge \int \frac{f_0(4y/R)}{|x-y|} \chi_\Omega(y) \d y  \\
&= \int \frac{f_0(4y/R)}{|x-y|} \chi_Z(y) \d y +  \int \frac{f_0(4y/R)}{|x-y|} (\chi_\Omega(y) - \chi_Z(y))  \d y \nn.
\end{align}
The last term of \eqref{eq:drop-binding-mainterm} can be estimated by using Lemma \ref{lem:chi-chiZ} with 
$f(y)=f_0(4y/R)|x-y|^{-1}$. By the triangle inequality, we see that for all $|x|\ge R$,
$$
\left|\nabla_y \Big(\frac{f_0(4y/R)}{|x-y|}\Big)\right| \le \frac{4|(\nabla f_0)(4y/R)|}{R|x-y|} + \frac{f_0(4y/R)}{|x-y|^2} \le \frac{C \1(2|y|\le R)}{R |x|},
$$
and hence
\bq \label{eq:gradient-fy}
\int \left|\nabla_y \Big(\frac{f_0(4y/R)}{|x-y|}\Big)\right|^2 \d y \le \frac{CR}{|x|^2}.
\eq
Therefore, by Lemma \ref{lem:chi-chiZ}, 
\bq \label{eq:drop-binding-mainterm-a}
\left| \int \frac{f_0(4y/R)}{|x-y|} (\chi_\Omega(y) - \chi_Z(y))  \d y \right| \le \frac{C\sqrt{R Q}}{|x|}, \quad \forall |x|\ge R.
\eq
On the other hand, since $f_0(4y/R) \chi_Z(y)$ is radially symmetric, we get
\begin{align} \label{eq:drop-binding-mainterm-b}
\int \frac{f_0(4y/R)}{|x-y|} \chi_Z(y) \d y  =  \int \frac{f_0(4y/R)}{|x|} \chi_Z(y) \d y = \frac{Z}{|x|}, \quad \forall |x| \ge R \ge 4R_Z
\end{align}
by Newton's theorem. Here recall that $R_Z=(3Z/(4\pi))^{1/3}$. Inserting \eqref{eq:drop-binding-mainterm-a} and \eqref{eq:drop-binding-mainterm-b} into \eqref{eq:drop-binding-mainterm}, we obtain  
\begin{align*}  \int_{|y| \le R} \frac{\chi_\Omega(y)}{|x-y|} \d y \ge \frac{Z-C\sqrt{R Q}}{|x|}, \quad \forall |x| \ge R\ge 4R_Z.
\end{align*}
Using this to estimate the left side of  \eqref{eq:key-estimate-in-proof} (and using $4\le C\sqrt{RQ}$ on the right side), we obtain  
$$
\frac{1}{2}|\Omega_{\ge R}|^2 \le C\sqrt{R Q} |\Omega_{\ge R}|+ 2 R \mathcal{H}^2 (\Omega_{=R}), \quad \forall R\ge 4R_Z.
$$
Consequently,
$$
|\Omega_{\ge R}| \le C \sqrt{R Q} + 2 \sqrt{R \mathcal{H}^2 (\Omega_{=R})}, \quad \forall R\ge 4R_Z.
$$
We can average over $[R,2R]$ to get
$$
|\Omega_{\ge 2R}| \le C \sqrt{R Q} + \frac{2}{R} \int_R^{2R} \sqrt{r \mathcal{H}^2 (\Omega_{=r})} dr, \quad \forall R\ge 4R_Z.
$$ 
By the Cauchy-Schwarz inequality and integration in spherical coordinates, 
\begin{align*}
\frac{1}{R}\int_{R}^{2R} \sqrt{r \mathcal{H}^2 (\Omega_{=r})} d r &\le \frac{1}{R}\left( \int_{R}^{2R}  r dr \right)^{1/2} \left( \int_{R}^{2R}  \mathcal{H}^2 (\Omega_{=r}) dr \right)^{1/2} \\
&= (3/2)^{1/2}  \Big( |\Omega_{\le 2R}| - |\Omega_{\le R}| \Big)^{1/2}.
\end{align*}
Therefore, 
$$
|\Omega_{\ge 2R}| \le C \sqrt{R Q} + C  \Big( |\Omega_{\le 2R}| - |\Omega_{\le R}| \Big)^{1/2},\quad \forall R\ge 4R_Z,
$$
which is equivalent to
\bq \label{eq:drop-N-final-0}
N \le |\Omega_{\le 2R}| + C \sqrt{RQ} + C  \Big( |\Omega_{\le 2R}| - |\Omega_{\le R}| \Big)^{1/2}, \quad \forall R\ge 4R_Z.
\eq

Finally, by Lemma \ref{lem:chi-chiZ} again, we have  
$$
\left| \int f_0(x/R) (\chi_\Omega (x)-\chi_Z(x))\d x \right| \le C \sqrt{R Q}, \quad \forall R>0.
$$
Consequently,
\bq \label{eq:Omega<R-upper}
|\Omega_{\le R}| \le \int f_0(x/R) \chi_\Omega (x) \d x \le Z + C\sqrt{RQ}, \quad \forall R>0,
\eq
and 
\bq \label{eq:Omega<R-lower}
|\Omega_{\le R}| \ge \int f_0(2x/R)\chi_\Omega(x) \d x  \ge  Z - C\sqrt{R Q}, \quad \forall R>0. 
\eq
Combining \eqref{eq:drop-N-final-0}, \eqref{eq:Omega<R-upper} and \eqref{eq:Omega<R-lower}, we conclude that
$$
N \le Z + C\sqrt{RQ} + C  (RQ)^{1/4}, \quad \forall R\ge 4R_Z.
$$
By choosing $R=4R_Z$, we obtain
$$
Q=N-Z \le C \sqrt{R_ZQ} + C  (R_Z Q)^{1/4}. 
$$
This implies $Q\le CR_Z \le C Z^{1/3}$ and completes the proof of Theorem \ref{thm:drop}. 
\end{proof}

\section{Exterior $L^1$-estimate} \label{sec:exterior-L1}

From now on we concentrate on TFDW theory. We always assume that $\rho_0$ is a minimizer for $E_Z^{\rm TFDW}(N)$ in \eqref{eq:variational-problem} with $N\ge Z$. We will denote by $\Phi_r(x)$ the screened nuclear potential in Theorem \ref{thm:screened-intro}. We also introduce the cut-off function 
$$
\chi_r^+(x) = \1(|x|\ge r).
$$
In this section, we control the number of electrons far from the nucleus. 
We start with the following simple observation.

\begin{lemma}[IMS-type formula] \label{lem:IMS} 
 For all smooth partition of unity $f_i : \R^3 \mapsto [0,1]$, $i = 1, \cdots, n$ such that $\sum_{i=1}^n f_i^2  =1$, $ \nabla f_i \in L^{\infty}$ and for all $\rho:\R^3 \to [0,\infty]$ such that $\sqrt{\rho} \in H^1(\R^3)$, we have
 \begin{align*}
& \sum_{i=1}^n \cE_{Z}^{\rm TFDW} (f_i^2 \rho)  -   \cE_{Z}^{\rm TFDW} (\rho) \\
& \qquad \qquad \le \sum_{i=1}^n \mathfrak{D}(f_i^2 \rho)- \mathfrak{D}(\rho) + C\Big(1 + \sum_{i=1}^n \norm{\nabla f_i}_{L^{\infty}}^2 \Big) \int_{A} \rho,
 \end{align*}
 where $A = \bigcup_{i=1}^n \{x \in \R^3 \, |\, 0< f_i(x)<1\}$.
\end{lemma}
\begin{proof}
 For the gradient term, we use the IMS formula 
 \begin{align*}
\sum_{i=1}^n \int |\nabla (f_i \sqrt{\rho})|^2  - \int |\nabla \sqrt{\rho}|^2  = \int \Big(\sum_{i=1}^n|\nabla f_i|^2 \Big) \rho \le \Big(\sum_{i=1}^n\norm{\nabla f_i}_{L^{\infty}}^2 \Big) \int_A \rho.
\end{align*}
For the Thomas-Fermi and Dirac terms, using
$$0\le 1- \sum_{i=1}^n f_i^{8/3} \le 1- \sum_{i=1}^n f_i^{10/3}  \le \1_A$$
and $c^{\rm D}\rho^{4/3}-c^{\rm TF}\rho^{5/3}\le C\rho$, we find that
\begin{align*}
&\int \Big(1- \sum_{i=1}^n f_i^{8/3}\Big) c^{\rm D}\rho^{4/3} -\Big(1-\sum_{i=1}^n f_i^{10/3}\Big) c^{\rm TF}\rho^{5/3}  \\
 &\qquad \le \int \Big(1-\sum_{i=1}^n f_i^{10/3}\Big) \big(c^{\rm D}\rho^{4/3}-c^{\rm TF}\rho^{5/3}\big) \le C \int_A \rho. \qedhere
\end{align*}
\end{proof}

Now we come to the main estimate of this section, which will allow us to control the TFDW minimizer $\rho_0$ in the exterior region.

\begin{lemma} \label{lem:binding} For all $r>0$, $s>0$ and $\lambda\in (0,1/2]$, we have
\begin{align*}
 \int\chi_r^+ \rho_0  &\le C \int_{r \le |x| \le (1+\lambda) r} \rho_0(x) \d x +  C \Big[\sup_{|z| \ge r} |z| \Phi_r(z) \Big]_+ \\
 & + C(\lambda^{-2}s^{-1}+s)+ Cs^{6/5}   \| \chi_{(1+\lambda) r}^+ \rho_0 \|_{L^{5/3}}.
\end{align*}
\end{lemma}

The main idea of the proof is similar to that of Lemma \ref{lem:drop-binding}. There is, however, a technical difference. In the liquid drop model, we could divide the minimizing set into two pieces by intersecting with $\{\nu\cdot x> \ell \}\cap \{ |x| \ge R\}$. In the TFDW theory, we have to use smeared out indicator functions of both the halfspace $\{\nu\cdot x> \ell \}$ and the outer set $\{ |x| \ge R\}$ in order to control the gradient term.  The scale of this smearing is set by $s$ and $\lambda$ in Lemma \ref{lem:binding}.

\begin{proof}[Proof of Lemma \ref{lem:binding}]  By the minimality of $\rho_0$, we have the binding inequality
\bq \label{eq:binding-inequality}
\cE_Z^{\rm TFDW}(\chi_1^2 \rho_0) + \cE^{\rm TFDW}_{Z=0} (\chi_2^2 \rho_0) -  \cE_Z^{\rm TFDW}(\rho_0) \ge 0
\eq
for any partition of unity $\chi_1^2+\chi_2^2=1$. For every $\ell>0, \nu\in \mathbb{S}^2$, we choose
$$
\chi_1 (x) = g_1 \Big( \frac{\nu \cdot \theta(x) -\ell}{s}\Big),\quad \chi_2(x)= g_2\Big( \frac{\nu \cdot \theta(x) -\ell}{s}\Big)
$$ 
where $g_1,g_2: \R \to \R$ and $\theta: \R^3\to \R^3$ satisfy 
$$
g_1^2+g_2^2=1, \quad g_1(t)=1 \text{~if~} t \le 0, \quad g_1(t)=0 \text{~if~} t \ge 1,\quad |g_1'| + |g_2'| \le C,
$$
$$
|\theta (x)| \le |x|, \quad \theta(x) =0 \text{~if~} |x| \le r, \quad \theta(x)=x \text{~if~} |x| \ge (1+\lambda) r, \quad |\nabla \theta|\le C \lambda^{-1}.
$$

Now let us bound the left side of \eqref{eq:binding-inequality} from above.
By Lemma~\ref{lem:IMS}, we have
\begin{align*}
& \cE_Z^{\rm TFDW}(\chi_1^2 \rho_0) + \cE^{\rm TFDW}_{Z=0} (\chi_2^2 \rho_0) -  \cE_Z^{\rm TFDW}(\rho_0) \\
& \qquad \qquad \le Z \int \frac{\chi_2^2(x) \rho_0(x)}{|x|} \d x  + \mathfrak{D}(\chi_1^2 \rho_0) + \mathfrak{D}(\chi_2^2 \rho_0)-\mathfrak{D}(\rho_0) \\
&  \qquad \qquad \qquad + C(1 + (\lambda s)^{-2}) \int_{\nu \cdot \theta(x) -s \le \ell \le \nu\cdot \theta(x)  } \rho_0 (x) \d x.
\end{align*}

For the attraction and interaction terms, we have
\begin{align*}
&\int \frac{Z\chi_2^2(x) \rho_0(x)}{|x|} \d x  + \mathfrak{D}(\chi_1^2 \rho_0) + \mathfrak{D}(\chi_2^2 \rho_0)-\mathfrak{D}(\rho_0) \\
&= \int \frac{Z\chi_2^2(x) \rho_0(x)}{|x|} \d x - \iint \frac{\chi_2^2(x) \rho_0(x) \chi_1^2(y) \rho_0(y)}{|x-y|} \d x \d y \\
&= \int \chi_2^2(x) \rho_0(x) \Phi_r (x) \d x - \iint_{|y| \ge r} \frac{\chi_2^2(x) \rho_0(x) \chi_1^2(y) \rho_0(y)}{|x-y|} \d x \d y \\
&\le \int_{\ell \le x \cdot \theta(x) } \rho_0(x) \big[\Phi_r(x)\big]_+ \d x - \iint_{|y| \ge r, \nu \cdot \theta(y) \le \ell \le \nu \cdot \theta(x)-s} \frac{ \rho_0(x) \rho_0(y)}{|x-y|} \d x \d y.
\end{align*}
Finally, since $\theta (x) = x$ when $|x| \ge (1+\lambda)r$,
\[
\iint_{\substack{|y| \ge r \\ \nu \cdot \theta(y) \le \ell \le \nu \cdot \theta(x)-s}} \frac{ \rho_0(x) \rho_0(y)}{|x-y|} \d x \d y
\ge \iint_{\substack{|x|,|y| \ge (1+\lambda) r \\ \nu \cdot y \le \ell \le \nu \cdot x-s}} \frac{ \rho_0(x) \rho_0(y)}{|x-y|} \d x \d y.
\]
In summary, the binding inequality \eqref{eq:binding-inequality} implies that
\begin{align} \label{eq:binding-consequence-1}
& C(1+(\lambda s)^{-2})  \int_{\nu \cdot \theta(x) -s \le \ell \le \nu\cdot \theta(x)  }  \rho_0(x) \d x+ \int_{\ell \le x \cdot \theta(x) } \rho_0(x) \big[\Phi_r(x)\big]_+ \d x  \nn\\
& \ge  \iint_{\substack{|x|,|y| \ge (1+\lambda) r \\ \nu \cdot y \le \ell \le \nu \cdot x-s}} \frac{ \rho_0(x) \rho_0(y)}{|x-y|} \d x \d y  
\end{align}
for all $\ell>0$ and $\nu \in \mathbb{S}^2$. Note that since $\ell>0$ and $\supp \theta \subset \{ |x| \ge r \}$, 
$$\int_{\nu \cdot \theta(x) -s \le \ell \le \nu\cdot \theta(x)  }  \rho_0(x) \d x = \int_{\nu \cdot \theta(x) -s \le \ell \le \nu\cdot \theta(x)  }  (\chi_r^+\rho_0)(x) \d x .$$

Integrating \eqref{eq:binding-consequence-1} over $\ell \in (0,\infty)$ we obtain
\begin{align*} 
& C(\lambda^{-2}s^{-1}+s)  \int \chi_r^+ \rho_0 + \int [\nu\cdot\theta(x)]_+ [\Phi_r(x)\big]_+ \rho_0(x) \d x \nn\\
& \ge \int_0^\infty \left( \iint_{\substack{|x|,|y| \ge (1+\lambda) r \\ \nu \cdot y \le \ell \le \nu \cdot x-s}} \frac{ \rho_0(x) \rho_0(y)}{|x-y|} \d x \d y \right) \d \ell
\end{align*}
Then we average over $\nu \in \mathbb{S}^2$ and use \eqref{eq:average-nu}. This gives 
\begin{align}  \label{eq:binding-consequence-2a}
& C(\lambda^{-2}s^{-1}+s)  \int \chi_r^+ \rho_0 + \frac{1}{4}\int |\theta(x)| [\Phi_r(x)\big]_+ \rho_0(x) \d x \nn\\
&\ge \int_{\mathbb{S}^2} \left( \int_0^\infty \left( \iint_{\substack{|x|,|y| \ge (1+\lambda) r \\ \nu \cdot y \le \ell \le \nu \cdot x-s}} \frac{ \rho_0(x) \rho_0(y)}{|x-y|} \d x \d y \right) \d \ell \right)\frac{d\nu}{4\pi} .
\end{align}
Using $|\theta(x)|\le |x| \chi_r^+$, we can estimate 
\begin{align*}
\int |\theta(x)| [\Phi_r(x)\big]_+ \rho_0(x) \d x \le \Big[\sup_{|z| \ge r} |z| \Phi_r(z) \Big]_+ \int \chi_r^+ \rho_0 .
\end{align*}
For the right side of \eqref{eq:binding-consequence-2a}, we write
\begin{align*}
& \int_{\mathbb{S}^2} \left( \int_0^\infty \left( \iint_{|x|, |y| \ge (1+\lambda)r} \1\Big( \nu \cdot y \le \ell \le \nu \cdot x-s \Big) \frac{ \rho_0(x) \rho_0(y)}{|x-y|} \d x \d y \right) \d \ell \right)\frac{d\nu}{4\pi} \nn\\
&= \int_{\mathbb{S}^2} \left( \int_0^\infty \left( \iint_{|x|, |y| \ge (1+\lambda)r} \left( \1\Big( \nu \cdot y \le \ell \le \nu \cdot x-s \Big) + \right. \right. \right.\nn\\
&\qquad\qquad\qquad\qquad\left. \left. \left. + \1\Big( - \nu \cdot x \le \ell \le - \nu \cdot y-s \Big) \right)  \frac{ \rho_0(x) \rho_0(y)}{|x-y|} \d x \d y \right) \d \ell \right)\frac{d\nu}{8\pi}.
\end{align*}
Then using Fubini's theorem and the elementary fact that 
\begin{align} \label{eq:binding-consequence-3b}
\int_0^\infty  \Big( \1\big(b \le \ell \le a - s\big) + \1\big(- a \le \ell \le - b - s\big) \Big) \d \ell \ge \Big[ [a-b]_+ -2s \Big]_+
\end{align}
with $a=\nu \cdot  x$ and $b=\nu \cdot y$, we obtain the lower bound
\begin{align*}
& \int_{\mathbb{S}^2} \left( \int_0^\infty \left( \iint_{\substack{|x|,|y| \ge (1+\lambda) r \\ \nu \cdot y \le \ell \le \nu \cdot x-s}} \frac{ \rho_0(x) \rho_0(y)}{|x-y|} \d x \d y \right) \d \ell \right)\frac{d\nu}{4\pi}\\
&\ge  \int_{\mathbb{S}^2} \left( \iint_{|x|,|y| \ge (1+\lambda)r} \left( \big[ \nu \cdot(x-y)\big]_+ -2s \right)  \frac{ \rho_0(x) \rho_0(y)}{|x-y|} \d x \d y \right)\frac{d\nu}{8\pi} \\
&= \iint_{|x|,|y| \ge (1+\lambda)r} \left( \frac{1}{8} - \frac{s}{|x-y|} \right) \rho_0(x) \rho_0(y)  \d x \d y \\
& = \frac{1}{8} \left( \int \chi_{(1+\lambda)r}^+ \rho_0 \right)^2 - 2s \mathfrak{D}(\chi_{(1+\lambda)r}^+ \rho_0). 
\end{align*}
Thus \eqref{eq:binding-consequence-2a} implies that
\begin{align} \label{eq:binding-consequence-4}
&\Big(C(\lambda^{-2} s^{-1}+s)+\frac{1}{4}\Big[\sup_{|z| \ge r} |z| \Phi_r(z) \Big]_+  \Big) \int \chi_{r}^+ \rho_0 \nn \\
&\ge \frac{1}{8} \left( \int \chi_{(1+\lambda)r}^+ \rho_0 \right)^2 - 2s \mathfrak{D}(\chi_{(1+\lambda)r}^+ \rho_0).
\end{align} 

We can add $(1/8)\Big(\int_{r \le |x| \le (1+\lambda)r} \rho_0\Big)^2$ to both sides of \eqref{eq:binding-consequence-4} and use
$$
\left( \int \chi_{(1+\lambda) r}^+ \rho_0 \right)^2  + \left(\int_{r \le |x| \le  (1+\lambda) r} \rho_0\right)^2 \ge \frac{1}{2} \left( \int \chi_{r}^+ \rho \right)^2.
$$
Moreover, by the Hardy-Littewood-Sobolev \cite[Theorem 4.3]{LieLos-01} and H\"older's inequalities,
\begin{align*} 
\mathfrak{D}(\chi_{(1+\lambda) r}^+ \rho_0) \leq C \| \chi_{(1+\lambda) r}^+ \rho_0 \|_{L^{6/5}}^2 \le  C \| \chi_{(1+\lambda) r}^+ \rho_0 \|_{L^{1}}^{7/6} \| \chi_{(1+\lambda) r}^+ \rho_0 \|_{L^{5/3}}^{5/6}.
\end{align*}
Thus \eqref{eq:binding-consequence-4} leads to 
\begin{align*} 
&\Big(C(\lambda^{-2} s^{-1}+s)+\frac{1}{4}\Big[\sup_{|z| \ge r} |z| \Phi_r(z) \Big]_+  \Big) \int \chi_{r}^+ \rho_0 + \frac{1}{8}\Big(\int_{r \le |x| \le (1+\lambda)r} \rho_0\Big)^2 \nn \\
&\qquad \ge \frac{1}{16} \left( \int \chi_{r}^+ \rho_0 \right)^2 - Cs \| \chi_{(1+\lambda) r}^+ \rho_0 \|_{L^{1}}^{7/6} \| \chi_{(1+\lambda) r}^+ \rho_0 \|_{L^{5/3}}^{5/6}.
\end{align*}
This implies the desired inequality. 
\end{proof}

As a by-product of the above proof, we have 

\begin{lemma} \label{lem:2Z}$ N\le 2Z + CZ^{2/3}+C.$
\end{lemma}
\begin{proof} In \eqref{eq:binding-consequence-4} we can use $|x| \Phi_r(x) \le Z$ and take $r\to 0^+$ . This gives 
\begin{align*}
\frac{1}{8} N^2 \le 2s \mathfrak{D}(\rho_0) + \frac{1}{4} Z N + C(\lambda^{-2}s^{-1}+s)  N
\end{align*}
for all $\lambda\in (0,1/2]$ and $s>0$. By choosing $\lambda=1/2$ and optimizing over $s>0$, we obtain 
\bq \label{eq:N<=2Z+D}
N \le 2Z + C\sqrt{(\mathfrak{D}(\rho_0)+N)N^{-1}}.
\eq
On the other hand, using the well-known lower bound on the TF ground state energy \cite{LieSim-77b} and the simple estimate $(c_{\rm TF}/2) \rho_0^{5/3} -c_{\rm D}\rho_0^{4/3} \ge - C\rho_0$, we find that
\begin{align} \label{eq:Drho0-rho05/3}
0 &\ge E^{\rm TFDW}_Z(N) = \cE^{\rm TFDW}_Z(\rho_0)\nn\\
&\ge \frac{c^{\rm TF}}{4} \int \rho_0^{5/3} + \frac{1}{2} \mathfrak{D}(\rho_0) + \Big(  \frac{c_{\rm TF}}{4} \int \rho^{5/3} + \int \frac{Z}{|x|} \rho_0 + \frac{1}{2}\mathfrak{D}(\rho_0) \Big) - C \int \rho_0 \nn\\
&\ge \frac{c^{\rm TF}}{4} \int \rho_0^{5/3} + \frac{1}{2} \mathfrak{D}(\rho_0)  - C(Z^{7/3}+N).
\end{align}
Thus $\mathfrak{D}(\rho_0)\le C(Z^{7/3}+N)$ and the conclusion follows from \eqref{eq:N<=2Z+D}.
\end{proof} 

Thus it remains to prove $N\le Z+C$ for $Z$ large. From now on, we will always assume that $Z\ge 1$. Note that when $Z$ becomes large, the bound in Lemma \ref{lem:2Z} is roughly twice of the desired bound. To improve this, we will only use Lemma \ref{lem:binding} to control the number of particles in the exterior region, where the number is small and losing a factor 2 is not a problem. The key observation is that the particles in the exterior region only feel the screened nuclear potential, which can be controlled by comparing with TF theory.  

In the rest of the paper, we will follow closely the strategy of \cite{Solovej-03}, but we also introduce some modifications and simplifications.

\section{Thomas-Fermi theory} \label{sec:TF}

In this section, we collect some useful facts from TF theory. We mostly follow \cite[Sections 4 \& 5]{Solovej-03}. First, we start with a general potential.

\begin{theorem} \label{thm:TF-V}(i) Let $V: \R^3 \to \R$ such that $V\in L^{5/2}+L^\infty$ and $V$ vanishes at infinity. For every $m>0$, there exists a unique minimizer $\rho_V^{\rm TF}$ for the TF energy functional
$$
\cE_V^{\rm TF}(\rho)= c^{\rm TF} \int \rho(x)^{5/3}\d x - \int V(x) \rho(x) \d x +\mathfrak{D}(\rho).
$$
subject to
$$ 
\rho\ge 0, \quad \rho\in L^{5/3}(\R^3) \cap L^1(\R^3), \quad \int \rho \le m.
$$
It satisfies the TF equation
$$
\frac{5c^{\rm TF}}{3}(\rho_V^{\rm TF}(x))^{2/3} = [\varphi_V^{\rm TF}(x) - \mu_V^{\rm TF}]_+
$$
with $\varphi_V^{\rm TF}(x)  = V(x)-\rho_V^{\rm TF}*|x|^{-1}$ and a constant $\mu_V^{\rm TF} \ge 0$.  Moreover, if $\mu_V^{\rm TF}>0$, then $\int \rho_V^{\rm TF}=m$.

(ii) Assume further that $V$ is harmonic for $|x|>r>0$, continuous for $|x|\ge r$ and $\lim_{|x|\to \infty} |x|V(x) \le m.$ If 
$$\mu_V^{\rm TF}< \inf_{|x|=r} \varphi_V^{\rm TF}(x),$$
then $\mu_V^{\rm TF}=0$, $\int \rho_V^{\rm TF} = \lim_{|x|\to \infty} |x|V(x)$ and for all $|x|>r$ we have  
$$
\varphi_V^{\rm TF}(x)> 0, \quad \Delta \varphi_V^{\rm TF} = 4\pi \rho_V^{\rm TF} = 4\pi \Big( \frac{3}{5 c^{\rm TF}} \Big)^{3/2}   \varphi_V^{\rm TF}(x)^{3/2}
$$ 
and the Sommerfeld estimate
\bq \label{eq:Sommerfeld}
\left(1+ a_r \Big(\frac{r}{|x|}\Big)^\zeta\right)^{-2} \le \frac{\varphi_V^{\rm TF}(x)}{A^{\rm TF} |x|^{-4}} \le1+ A_r \Big(\frac{r}{|x|}\Big)^\zeta.
\eq
Here $A^{\rm TF}=(5 c^{\rm TF})^3(3 \pi^2)^{-1}$, $\zeta=(\sqrt{73}-7)/2 \approx 0.77$ and
$$ 
a_r=\sup_{|z|=r} \left(   \frac{\varphi_V^{\rm TF}(z)}{A^{\rm TF} |z|^{-4}} \right)^{-1/2} -1, \quad A_r=  \sup_{|z|=r} \frac{\varphi_V^{\rm TF}(z)}{A^{\rm TF} |z|^{-4}}-1.
$$
\end{theorem}

\begin{proof} Part (i) is well-known from Lieb and Simon \cite{LieSim-77b}. Part (ii) essentially follows from \cite[Section 4]{Solovej-03}. More precisely, if $\mu_V^{\rm TF}>0$, then from \cite[Corollary 4.7]{Solovej-03} we have 
$$\int \rho_V^{\rm TF}<\limsup_{|x|\to \infty} |x|V(x)\le m$$
but this contradicts with the last statement in part (i). Thus $\mu_V^{\rm TF}=0$. From \cite[Theorem 4.3 and Corollary 4.8]{Solovej-03} we obtain $\int \rho_V^{\rm TF} = \lim_{|x|\to \infty} |x|V(x)$. Since $\varphi_V^{\rm TF}$ is harmonic for $|x|> r$, vanishes at infinity and $\inf_{|x|=r} \varphi_V^{\rm TF}(x)>0$, we have $\varphi_V^{\rm TF}(x)>0$ for all $|x|\ge r$ by the strong maximum principle. The bound on $\varphi_V^{\rm TF}/(A^{\rm TF}|x|^{-4})$ follows from \cite[Lemma 4.4]{Solovej-03}.
\end{proof}

In particular, for the standard atomic case $V=Z/|x|$, we have
\begin{theorem} \label{thm:TF-Z} The TF energy functional
$$
\cE^{\rm TF}(\rho)= c^{\rm TF}\int \rho(x)^{5/3} \d x - \int \frac{Z\rho(x)}{|x|}  \d x  +\mathfrak{D}(\rho)
$$
has a unique minimizer $\rho^{\rm TF}$ over all $0\le \rho \in L^{5/3}(\R^3) \cap L^1(\R^3).$ We have $\int \rho^{\rm TF} =Z$. Moreover, 
$$0< \varphi^{\rm TF}(x) = Z|x|^{-1} - \rho^{\rm TF}*|x|^{-1} \le A^{\rm TF} |x|^{-4},\quad \forall |x|>0$$
and 
$$\varphi^{\rm TF}(x) \ge A^{\rm TF} |x|^{-4} \left(1+ C \Big(\frac{Z^{-1/3}}{|x|}\Big)^\zeta\right)^{-2}, \quad \forall |x|\ge Z^{-1/3}.$$
Consequently,
$$\rho^{\rm TF}(x) \le \Big(\frac{3 A^{\rm TF}}{5c^{\rm TF}}\Big)^{3/2} |x|^{-6},\quad \forall |x|>0$$
and 
$$\rho^{\rm TF}(x) \ge \Big(\frac{3 A^{\rm TF}}{5 c^{\rm TF}}\Big)^{3/2} |x|^{-6} \left(1+ C \Big(\frac{Z^{-1/3}}{|x|}\Big)^\zeta\right)^{-3}, \quad \forall |x|\ge Z^{-1/3}.$$
\end{theorem}

\begin{proof} Since $V(x)=Z/|x|$ is harmonic for $|x|>0$, we can apply Theorem \ref{thm:TF-V} (ii) for every $r>0$. The condition $\mu^{\rm TF}<\inf_{|x|=r}\varphi^{\rm TF}(x)$ holds true for $r>0$ small enough because $|x|\varphi^{\rm TF}(x)\to Z$ as $|x|\to 0$. Thus $\mu^{\rm TF}=0$ and $\int \rho^{\rm TF}=Z$ (as $|x|V(x)=Z$).  

Now we bound $\varphi^{\rm TF}$ using \eqref{eq:Sommerfeld}. Since $|x|\varphi^{\rm TF}(x)\to Z$ as $|x|\to 0$, we have $A_r\to -1$ as $r\to 0$, and hence $\varphi^{\rm TF}(x) \le A^{\rm TF} |x|^{-4}$ for all $|x|>0$. On the other hand, note that
$$
(\rho^{\rm TF}*|.|^{-1})(y) \le \int \frac{\rho^{\rm TF}}{|x|} \d x \le CZ^{4/3}, \quad \forall y.
$$
Here the first estimate follows from $\cE^{\rm TF}(\rho^{\rm TF})\le \cE^{\rm TF}(\rho^{\rm TF}(.-y))$ and the second estimate is a consequence of the well-known fact that the TF ground state is $Z^2$ times a universal ($Z$-independent) function of the variable $Z^{-1/3} x$. Therefore, by choosing 
$$R=\beta_0 Z^{-1/3}$$
with a universal constant $\beta_0 \in (0,1)$ which is sufficiently small, we have 
$$
\varphi^{\rm TF}(x) \ge \frac{Z}{|x|} - CZ^{4/3} \ge C^{-1}R^{-4}, \quad \forall |x|=R.
$$
Applying the lower bound in \eqref{eq:Sommerfeld} with $r=R$, we obtain   
$$\varphi^{\rm TF}(x) \ge A^{\rm TF} |x|^{-4} \left(1+ C \Big(\frac{Z^{-1/3}}{|x|}\Big)^\zeta\right)^{-2}, \quad \forall |x|\ge R. 
$$
The bounds on $\rho^{\rm TF}$ follow from the bounds on $\varphi^{\rm TF}$ and the TF equation.
\end{proof}

\section{Splitting outside from inside} \label{sec:split}

In this section, we will split the energy from the interior region and the exterior region, in the spirit of \cite[Section 6]{Solovej-03}. To make the idea transparent, let us warm up with standard TF theory. Recall that $\chi_r^+=\1(|x| \ge r)$ and we continue using the notations from Theorem \ref{thm:TF-Z}.

\begin{lemma} \label{lem:outside-energy-TF} For every $r>0$, we have 
$$
\widetilde{\cE}_r(\chi_r^+\rho^{\rm TF}) \le \widetilde{\cE}_r(\rho)
$$ 
for all  $0\le \rho \in L^{5/3}(\R^3)\cap L^1(\R^3)$ with $\supp \rho \subset \{|x| \ge r\}$, where
$$
\widetilde{\cE}_r(\rho)= c^{\rm TF}\int \rho^{5/3} - \int \Phi_r^{\rm TF} \rho + \mathfrak{D}(\rho), \quad \Phi_r^{\rm TF}(x) = \frac{Z}{|x|} -\int_{|y|<r} \frac{\rho^{\rm TF}(y)}{|x-y|} \d y.
$$
\end{lemma}

\begin{proof} For all $0\le \rho \in L^{5/3}(\R^3)\cap L^1(\R^3)$ with $\supp \rho \subset \{|x| \ge r\}$, by the minimality of $\rho^{\rm TF}$, we have
$$
\cE^{\rm TF}(\rho^{\rm TF}) \le \cE^{\rm TF}(\chi_r^-\rho^{\rm TF}+\rho) 
$$
where $\chi_r^-=\1(|x|<r)$. Since $\chi_r^-\rho^{\rm TF}$ and $\rho$ have disjoint supports, we can write
\begin{align*}
\cE^{\rm TF}(\chi_r^-\rho^{\rm TF}+\rho)&=\cE^{\rm TF}(\chi_r^-\rho^{\rm TF}) + \cE^{\rm TF}(\rho) + \iint \frac{\rho(x) (\chi_r^-\rho^{\rm TF})(y)}{|x-y|} \d x \d y \\
& = \cE^{\rm TF}(\chi_r^-\rho^{\rm TF}) + \widetilde {\cE}_r(\rho).
\end{align*}
In particular, we can apply the latter equality with $\rho=\chi_r^+ \rho^{\rm TF}$ and obtain 
$$
\cE^{\rm TF}(\rho^{\rm TF}) = \cE^{\rm TF}(\chi_r^-\rho^{\rm TF}+\chi_r^+ \rho^{\rm TF}) =  \cE^{\rm TF}(\chi_r^-\rho^{\rm TF}) + \widetilde {\cE}_r(\chi_r^+\rho^{\rm TF}). 
$$
Thus 
\begin{equation*} 0 \le \cE^{\rm TF}(\chi_r^-\rho^{\rm TF}+\rho) -  \cE^{\rm TF}(\rho^{\rm TF}) = \widetilde {\cE}_r(\rho) - \widetilde {\cE}_r(\chi_r^+\rho^{\rm TF}). \qedhere
\end{equation*}
\end{proof}

Now we prove an analogue of Lemma \ref{lem:outside-energy-TF} for TFDW theory. Because of the gradient term, we cannot take the cut-off $\chi_r^+$ directly. Instead, for $\lambda\in (0,1/2]$, let us introduce a partition of unity
$$
\eta_-^2 + \eta_{(0)}^2 + \eta_r^2 =1
$$
such that 
$$ \supp \eta_r \subset \{|x| > r\}, \quad \eta_r(x) =1\text{~if~} |x| \ge (1+\lambda)r, $$
$$ \supp \eta_- \subset \{|x| < r\}, \quad \eta_-(x) =1 \text{~if~} |x| \le (1-\lambda) r,$$
$$|\nabla \eta_-|^2 +|\nabla \eta_{(0)}|^2+|\nabla \eta_r|^2 \le C(\lambda r)^{-2}.$$
We have
\begin{lemma} \label{lem:outside-energy}For all $r>0$ and $\lambda\in (0,1/2]$, we have
$$
\cE_r^{\rm A}(\eta_r^2 \rho_0) \le \cE_r^{\rm A}(\rho) + \mathcal{R} 
$$
for all $ 0\le \rho \in L^{5/3}(\R^3)\cap L^1(\R^3)$ with $\supp \rho \subset \{|x| \ge r\}$ and $\int \rho \le \int \chi_r^+ \rho_0$, where we have introduced the auxiliary functional 
\begin{align*}
\cE_r^{\rm A}(\rho) &= c^{\rm TF} \int \rho^{5/3} +c^{\rm W}\int |\nabla \sqrt{\rho}|^2 - \int \Phi_r \rho + \mathfrak{D}(\rho, \rho) 
\end{align*}
and
$$
\mathcal{R} \le  C(1+(\lambda r)^{-2}) \int_{|x| \ge (1-\lambda)r} \rho_0 + C \lambda r^3 \Big[ \sup_{|z| \ge (1-\lambda)r}  \Phi_{|z|}(z) \Big]_+^{5/2} + C\int (\eta_r^2 \rho_0)^{4/3}.$$ \end{lemma}


\begin{proof} The proof is similar to \cite[Theorem 6.2]{Solovej-03}. We will show that
\bq \label{eq:lem-eta+1-a}
\cE_r^{\rm A}(\eta_r^2 \rho_0) + \cE^{\rm TFDW}(\eta_-^2\rho_0) - \mathcal{R} \le \cE^{\rm TFDW}(\rho_0) \le \cE_r^{\rm A}(\rho)+ \cE^{\rm TFDW}(\eta_-^2 \rho_0)
\eq
for all 
$$\rho\ge 0, \quad \supp \rho \subset \{|x| \ge r\}, \quad \int \rho \le \int_{|x|\ge r} \rho_0.$$

{\bf Upper bound.} Note that $\eta_-^2 \rho_0$ and $\rho$ have disjoint supports and $\int (\eta_-^2 \rho_0 +\rho) \le N$. 
By the minimality of $\rho_0$ and the fact that $N\mapsto E_{\rm TFDW}(N)$ is nonincreasing (see \cite[Theorem 8.4]{Lieb-81b}), we have
\begin{align*}
\cE^{\rm TFDW}(\rho_0) & \le \cE^{\rm TFDW}(\eta_-^2 \rho_0 + \rho) \\
&=  \cE^{\rm TFDW}(\eta_-^2 \rho_0)+ \cE^{\rm TFDW}(\rho) + \iint \frac{(\eta_-^2\rho_0)(x) \rho(y)}{|x-y|} \d x \d y \\
& \le \cE^{\rm TFDW}(\eta_-^2 \rho_0) + \cE^{\rm TFDW} (\rho)  + \iint_{|x| \le r} \frac{\rho_0(x) \rho(y)}{|x-y|} \d x \d y \\
& = \cE^{\rm TFDW}(\eta_-^2 \rho_0)  + \cE_r^{\rm A}(\rho)- c^{\rm D}\int \rho^{4/3}.
\end{align*} 

{\bf Lower bound.}
We apply Lemma~\ref{lem:IMS} to see that
\begin{align*}
\cE^{\rm TFDW} & (\rho_0)  \ge \cE^{\rm TFDW}(\eta_-^2 \rho_0) + \cE^{\rm TFDW}(\eta_{(0)}^2 \rho_0) + \cE^{\rm TFDW} (\eta_r^2 \rho_0) \\
 - & C(1+(\lambda r)^{-2}) \int_{(1+\lambda)r \ge |x| \ge (1-\lambda)r} \rho_0 + \iint \frac{(\eta_-^2 \rho_0)(x) (\eta_{(0)}^2 \rho_0)(y)}{|x-y|} \d x \d y\\
 + &  \iint \frac{((\eta_-^2 + \eta_{(0)}^2) \rho_0)(x) (\eta_r^2 \rho_0)(y)}{|x-y|} \d x \d y.
\end{align*}
We have
\begin{align*}
&\cE^{\rm TFDW}(\eta_{(0)}^2 \rho_0)  +   \iint \frac{(\eta_-^2 \rho_0)(x) (\eta_{(0)}^2 \rho_0)(y)}{|x-y|} \d x \d y  \\
& \ge c^{\rm TF}\int (\eta_{(0)}^2 \rho_0)^{5/3} - \int \frac{Z}{|x|} (\eta_{(0)}^2\rho_0) \\
&\qquad \qquad\qquad + \iint_{|x| \le (1-\lambda) r} \frac{\rho_0(x) (\eta_{(0)}^2 \rho_0)(y)}{|x-y|} -  c^{\rm D}\int (\eta_{(0)}^2 \rho_0)^{4/3}   \\
& = \int \Big[ c^{\rm TF}(\eta_{(0)}^2 \rho_0)(x)^{5/3} - \Phi_{(1-\lambda)r}(x)  (\eta_{(0)}^2 \rho_0)(x) -  c^{\rm D}(\eta_{(0)}^2 \rho_0)(x)^{4/3} \Big] \d x   \\
&\ge -C \int_{\supp \eta_{(0)}} \Big( [\Phi_{(1-\lambda)r}(x)]_+^{5/2} + \rho_0 \Big) \d x  \\
&\ge - C \lambda r^3 \Big[ \sup_{|z| \ge (1-\lambda)r}  \Phi_{(1-\lambda)r}(z) \Big]_+^{5/2} - C \int_{(1+\lambda)r \ge |x| \ge (1-\lambda)r} \rho_0.
\end{align*}
In the last estimate we have used 
$$\supp \eta_{(0)} \subset \{ (1+\lambda)r \ge |x| \ge (1-\lambda)r\}.$$
Moreover,
\begin{align*}
&\cE^{\rm TFDW} (\eta_r^2 \rho_0) +  \iint \frac{((\eta_-^2 + \eta_{(0)}^2) \rho_0)(x) (\eta_r^2 \rho_0)(y)}{|x-y|} \d x \d y \\
& \ge \cE^{\rm TFDW} (\eta_r^2 \rho_0) +  \iint_{|x|\le r} \frac{\rho_0(x) (\eta_r^2 \rho_0)(y)}{|x-y|} \d x \d y \\
&= \cE_r^{\rm A} (\eta_r^2 \rho_0) - c^{\rm D}\int (\eta_r^2 \rho_0)^{4/3}.
\end{align*} 

In summary, we have proved that
$$
\cE^{\rm TFDW}(\rho_0)  \ge \cE^{\rm TFDW}(\eta_-^2 \rho_0) + \cE_r^{\rm A} (\eta_r^2 \rho_0) - \mathcal{R} 
$$
where
\begin{align*}
\mathcal{R} &\le  C(1+(\lambda r)^{-2}) \int_{(1+\lambda)r \ge |x| \ge (1-\lambda)r} \rho_0 \\
&\qquad \qquad + C \lambda r^3 \Big[ \sup_{|z|\ge (1-\lambda)r}  \Phi_{(1-\lambda)r}(z) \Big]_+^{5/2} + C\int (\eta_r^2 \rho_0)^{4/3}.
\end{align*}
Thus \eqref{eq:lem-eta+1-a} holds true and this ends the proof.
\end{proof}

As a simple consequence of Lemma \ref{lem:outside-energy}, we have  
\begin{lemma} \label{lem:outside-kinetic} For all $r>0$ and $\lambda\in (0,1/2]$,
\begin{align*}
\int \chi_{(1+\lambda)r}^+ \rho_0^{5/3} &\le C(1+ (\lambda r)^{-2}) \int_{|x| \ge (1-\lambda)r} \rho_0  \\
&+  C \lambda r^3 \Big[ \sup_{|z| \ge (1-\lambda)r}  \Phi_{(1-\lambda)r}(z) \Big]_+^{5/2} + C \Big[\sup_{|z| \ge r} |z| \Phi_{r}(z)\Big]_+^{7/3} .
\end{align*}
\end{lemma}

\begin{proof} We have $\cE_r^{\rm A}(\eta_r^2 \rho_0) \le \mathcal{R}$ by Lemma \ref{lem:outside-energy}  (we can choose $\rho=0$). On the other hand, by the ground state energy in TF theory, 
\begin{align*}
\cE_r^{\rm A}(\eta_r^2 \rho_0) &\ge c^{\rm TF} \int (\eta_r^2 \rho_0)^{5/3} - \Big[\sup_{|z|\ge r} |z| \Phi_{r}(z)\Big]_+ \int \frac{\eta_r^2 \rho_0}{|x|} \d x + \mathfrak{D}(\eta_r^2 \rho_0) \\
& \ge \frac{c^{\rm TF}}{2}  \int (\eta_r^2 \rho_0)^{5/3} - C \Big[\sup_{|z| \ge r} |z| \Phi_{r}(z)\Big]_+^{7/3}.
\end{align*} 
Therefore,
\bq \label{eq:lem:outside-kinetic-proof-a}
 \int (\eta_r^2 \rho_0)^{5/3} \le C \mathcal{R} + C \Big[\sup_{|z| \ge r} |z| \Phi_{r}(z)\Big]_+^{7/3}.
\eq
If we insert the bound on $\mathcal{R}$ in Lemma \ref{lem:outside-energy} into \eqref{eq:lem:outside-kinetic-proof-a} and use 
$$
\int (\eta_r^2\rho_0)^{4/3} \le \left( \int \eta_r^2\rho_0 \right)^{1/2} \left( \int (\eta_r^2\rho_0)^{5/3} \right)^{1/2},
$$ 
then we obtain 
\begin{align*}
\int (\eta_r^2 \rho_0)^{5/3} &\le C(1+ (\lambda r)^{-2}) \int_{|x| \ge (1-\lambda)r} \rho_0  \\
&+  C \lambda r^3 \Big[ \sup_{|z| \ge (1-\lambda)r}  \Phi_{(1-\lambda)r}(z) \Big]_+^{5/2} + C \Big[\sup_{|z| \ge r} |z| \Phi_{r}(z)\Big]_+^{7/3} .
\end{align*}
Since $\eta_r^2\ge \chi_{(1+\lambda)r}^+$, the desired inequality follows.
\end{proof}

\section{Screened potential estimate} \label{sec:bootstrap}

Recall that we are always assuming $Z\ge 1$. Our main result in this section is the following

\begin{lemma}[Screened potential estimate] \label{lem:screened} There are universal constants $C>0,\eps>0,D>0$ such that 
$$
| \Phi_{|x|}(x)  - \Phi^{\rm TF}_{|x|}(x)  | \le C |x|^{-4+\eps}, \quad \forall |x|\le D.
$$
\end{lemma}

This is the main technical tool of the whole approach.  We prove Lemma \ref{lem:screened} using a bootstrap argument based on two lemmas.   
\begin{lemma}[Initial step] \label{thm:screened-first} There is a universal constant $C_1>0$ such that 
$$
| \Phi_{|x|}(x)  - \Phi^{\rm TF}_{|x|}(x)  | \le C_1 Z^{4/3}|x|^{1/12}, \quad \forall |x|>0.
$$
\end{lemma}

\begin{lemma}[Iterative step] \label{thm:screened-it} There are universal constants $C_2, \beta,\delta,\eps >0$ such that, if 
\bq \label{eq:assume-D}
| \Phi_{|x|}(x)  - \Phi^{\rm TF}_{|x|}(x)  | \le \beta |x|^{-4}, \quad \forall |x| \le D
\eq
for some $D\in [Z^{-1/3}, 1]$, then  
\bq \label{eq:assume-D-it}
| \Phi_{|x|}(x)  - \Phi^{\rm TF}_{|x|}(x)  | \le C_2 |x|^{-4+\eps}, \quad \forall D\le |x| \le D^{1-\delta}.
\eq
\end{lemma}

We will prove Lemmas \ref{thm:screened-first} and \ref{thm:screened-it} later. Now let us provide 

\begin{proof}[Proof of Lemma \ref{lem:screened}] We use the notations in Lemmas \ref{thm:screened-first} and \ref{thm:screened-it} and set $\sigma=\max\{C_1,C_2\}$. Without loss of generality we may assume that $\beta<\sigma$ and $\eps\le 1/12$. Let us denote
$$D_n=Z^{-\frac{1}{3}(1-\delta)^n}, \quad n=0,1,2,...$$
From Lemma \ref{thm:screened-first}, we have
$$
|\Phi_{|x|}(x)-\Phi_{|x|}^{\rm TF}(x)|\le C_1 Z^{4/3}|x|^{1/12} \le \sigma |x|^{-4+\eps}, \quad \forall |x| \le D_0=Z^{-1/3}.
$$
From Lemma \ref{thm:screened-it}, we deduce by induction that for all $n=0,1,2,...$, if 
$$\sigma (D_n)^\eps \le \beta,$$
then 
$$ |\Phi_{|x|}(x)-\Phi_{|x|}^{\rm TF}(x)|\le \sigma |x|^{-4+\eps}, \quad \forall |x|\le D_{n+1}.$$

Note that $D_n\to 1$ as $n\to\infty$ and that $\sigma>\beta$. Thus, there is a minimal $n_0=0,1,2,\ldots$ such that $\sigma (D_{n_0})^\epsilon>\beta$. If $n_0\ge 1$, then $\sigma (D_{n_0-1})^\epsilon\le\beta$ and therefore by the preceding argument
$$
|\Phi_{|x|}(x)-\Phi_{|x|}^{\rm TF}(x)|\le \sigma |x|^{-4+\eps}, \quad \forall |x|\le D_{n_0} \,.
$$
As we have already shown, the same bound holds for $n_0=0$. Let $D =(\sigma^{-1}\beta)^{1/\epsilon}$, which is a universal constant, and note that by choice of $n_0$ we have $D_{n_0}\geq D$. This proves the lemma.
\end{proof}

It remains to prove Lemmas \ref{thm:screened-first} and \ref{thm:screened-it}.

\subsection{Coulomb estimates}

Before proving Lemmas \ref{thm:screened-first} and \ref{thm:screened-it}, let us recall some general useful facts. 

From \cite[Corollary 9.3]{Solovej-03}, we have the following consequences of the Feffer\-man-Seco inequality \eqref{eq:Fef-Sec}.

\begin{lemma} \label{lem:f*1/|x|} For every $f\in L^{5/3} \cap L^{6/5} (\R^3)$ and $x\in \R^3$, we have 
\bq \label{eq:Coulomb-estimate-1}
(f*|.|^{-1})(x)\le C \|f\|_{L^{5/3}}^{5/7} (\mathfrak{D}(f))^{1/7}
\eq
and
\bq \label{eq:Coulomb-estimate-2}
\int_{|y|<|x|} \frac{f(y)}{|x-y|} \d y \le C \|f\|_{L^{5/3}}^{5/6} (|x|\mathfrak{D}(f))^{1/12}.
\eq
\end{lemma}

\begin{proof} From \cite[Eq. (82)]{Solovej-03} we have
$$
(f*|.|^{-1})(x) \le C s^{1/5} \|f\|_{L^{5/3}} + Cs^{-1/2} \sqrt{\mathfrak{D}(f)}, \quad \forall s>0
$$
(the statement of \cite[Eq. (82)]{Solovej-03} has a typo where $s^{1/5}$ appears twice, but the correct estimate can be found in the proof). Then optimizing over $s$ we obtain \eqref{eq:Coulomb-estimate-1}.

By \cite[Eq. (83)]{Solovej-03} (with $\kappa=s/|x|$),  
$$
\int_{|y|<|x|} \frac{f(y)}{|x-y|} \d y \le C s^{1/5} \|f\|_{L^{5/3}} + Cs^{-1}|x|^{1/2} \sqrt{\mathfrak{D}(f)}, \quad \forall s>0.
$$
Optimizing over $s>0$ leads to \eqref{eq:Coulomb-estimate-2}.
\end{proof}

In the iterative step we will use the maximum principle. Note that $\Delta \Phi_r(x)=4\pi \1(|x|\le r)\rho_0(x)$ in the distributional sense, and hence $\Phi_r^{\rm TF}(x)$ is subharmonic for $|x|>0$ and harmonic for $|x|>r$. We have the following well-known fact for subharmonic functions.  
\begin{lemma} \label{lem:harmonic} Let $f: \{x\in \R^3: |x| \ge r\} \to \R$. Assume that $f$ is subharmonic for $|x|>r$ (namely $\Delta f \ge 0$ in the distributional sense), continuous for $|x| \ge r$ and vanishing at infinity. Then 
$$ \sup_{|x| \ge r} |x|f(x) = \sup_{|x| = r} |x|f(x) .$$
\end{lemma}
\begin{proof} Let $g(x)=f(x)- F_r/|x|$ with $F_r=r\sup_{|z|=r} f(z)$. Since $g$ is subharmonic for $|x|>r$, by the maximum principle we have
$$ \sup_{|x|\ge r}g(x) = \max \Big\{ \sup_{|x|=r} g(x), \limsup_{|x|\to \infty} g(x)\Big\}=0.$$
Therefore,
$$
\sup_{|x| \ge r} |x|f(x) = \sup_{|x| \ge r} (|x|g(x) + F_r) \le F_r =\sup_{|x| = r} |x|f(x)
$$ 
and the conclusion follows.
\end{proof}

\subsection{Initial step}
\begin{proof}[Proof of Lemma \ref{thm:screened-first}] By the variational principle and the well-studied Thomas-Fermi-von Weis\"acker theory in \cite[Theorem 7.30]{Lieb-81b}, we have for all $Z\ge 1$, 
\begin{align*}
\cE^{\rm TFDW}(\rho_0) &\le \inf \Big\{ \cE^{\rm TFDW}(\rho)+c^{\rm D}\int \rho^{4/3}: \rho\ge 0, \sqrt{\rho}\in H^1(\R^3) \Big\} \\
&\le \cE^{\rm TF}(\rho^{\rm TF}) + CZ^2.
\end{align*}
On the other hand, we have $\int \rho_0 = N\le CZ$ by Lemma \ref{lem:2Z}. We also have $\int \rho_0^{5/3} \le C Z^{7/3}$ by \eqref{eq:Drho0-rho05/3}. Thus by H\"older's inequality 
$$
\int \rho_0^{4/3} \le \left( \int \rho_0 \right)^{1/2}\left( \int \rho_0^{5/3} \right)^{1/2} \le C Z^{5/3}.
$$
Consequently,
$$
\cE^{\rm TFDW}(\rho_0) \ge \cE^{\rm TF}(\rho_0) - C Z^{5/3}.
$$ 
In summary, we have
$$
\cE^{\rm TF}(\rho_0) \le \cE^{\rm TF}(\rho^{\rm TF}) + CZ^2.
$$
By the minimality of $\rho^{\rm TF}$ and the convexity of $\rho^{5/3}$ we have 
\begin{align} \label{eq:cv-D-TF}
C Z^2 &\ge \cE^{\rm TF}(\rho_0) +  \cE^{\rm TF}(\rho^{\rm TF}) - 2 \cE^{\rm TF}\Big(\frac{\rho_0+\rho^{\rm TF}}{2}\Big) \nn\\
& \ge \mathfrak{D}(\rho_0)+\mathfrak{D}(\rho^{\rm TF}) - 2\mathfrak{D}\Big(\frac{\rho_0+\rho^{\rm TF}}{2}\Big) =\frac{1}{2} \mathfrak{D}(\rho_0- \rho^{\rm TF}).
\end{align}
Now we use the Coulomb estimate \eqref{eq:Coulomb-estimate-2} with $f=\pm (\rho_0-\rho^{\rm TF})$ together with \eqref{eq:cv-D-TF} and the kinetic estimates $\int \rho_0^{5/3} \le CZ^{7/3}$, $\int (\rho^{\rm TF})^{5/3} \le CZ^{7/3}$. We have for all $|x|>0$,
\begin{align*}
|\Phi_{|x|}(x)-\Phi_{|x|}^{\rm TF}(x)| &=\left| \int_{|y|<|x|} \frac{\rho_0(y) - \rho^{\rm TF}(y)}{|x-y|} \d y \right| \\
&\le C \|\rho_0-\rho^{\rm TF}\|_{L^{5/3}}^{5/6} (|x|\mathfrak{D}(\rho_0-\rho^{\rm TF}))^{1/12} \le CZ^{4/3}|x|^{1/12}.
\end{align*}
This finishes the proof.
\end{proof}

\subsection{Iterative step}

Now we prove Lemma \ref{thm:screened-it}. 

Let us summarize the overall idea for the reader's convenience. As in \cite[Eq. (97)]{Solovej-03}, when $|x|\ge r$  we can decompose 
\begin{align} \label{eq:rev-decom}
\Phi_{|x|}(x)-\Phi_{|x|}^{\rm TF}(x) &= \varphi_r^{\rm TF}(x) -\varphi^{\rm TF}(x)+ \int_{|y|>|x|} \frac{\rho_r^{\rm TF}(y)-\rho^{\rm TF}(y)}{|x-y|} \d y,\nn\\
&\qquad + \int_{|y|<|x|}  \frac{\rho_r^{\rm TF}(y)-(\chi_r^+\rho_0)(y)}{|x-y|} \d y
\end{align}
where $\rho_r^{\rm TF}$ is the minimizer of the exterior TF functional associated with the screened potential $V_r=\chi_r^+ \Phi_r$ (recall that $\chi_r^+=\1(|x|\ge r)$) and $\varphi_r^{\rm TF}(x)= V_r(x)-\rho_r^{\rm TF}*|x|^{-1}$. Then we bound $\rho_r^{\rm TF}-\rho^{\rm TF}$ and $\varphi_r^{\rm TF} -\varphi^{\rm TF}$  using the Sommerfeld estimate \eqref{eq:Sommerfeld} and bound $\rho_r^{\rm TF}-\chi_r^+\rho_0$ using the energy estimate in Lemma \ref{lem:outside-energy}. Optimizing these bounds over $r\in (0,D]$ leads to the desired result \eqref{eq:assume-D-it}. The role of the assumption \eqref{eq:assume-D} is to provide a-priori estimates for $\rho_0$ in the outer region $\{|x|\ge r\}$.

Now we go to the details. The proof is divided into several steps. Recall that we always denote by $C$ a universal constant (in particular, it is independent of $N,Z,\beta,D$). \\

\noindent
{\bf Step 1.} We collect some easy consequences of \eqref{eq:assume-D}. 

\begin{lemma} \label{lem:screened-easy-bounds} 
Assume that \eqref{eq:assume-D} holds true for some $\beta, D\in (0,1]$. Then for all $r\in (0, D]$, we have
\begin{align}\label{eq:int-rho-1}
\left| \int_{|x|<r} (\rho_0 - \rho^{\rm TF}) \right|\le \beta r^{-3},\\
\label{eq:int-rho-2}
\sup_{|x| \ge r} |x| \Phi_{r}(x) \le C r^{-3},\\
\int_{|x|>r} \rho_0^{5/3} \le C r^{-7},\label{eq:int-rho-3}\\
\int_{|x|>r} \rho_0 \le C r^{-3}.\label{eq:int-rho-4}
\end{align}
\end{lemma}  

\begin{proof} Let $r\in (0,D]$. By Newton's theorem, we have
\begin{align*} 
\int_{|y|<r} (\rho^{\rm TF}(y)-\rho_0(y)) \d y &= r \int_{\mathbb{S}^2} \left( \int_{|y|<r} \frac{\rho^{\rm TF}(y)-\rho_0(y) }{|r\nu - y|} \d y \right) \frac{\d \nu}{4\pi} \\
&= r \int_{\mathbb{S}^2} (\Phi_r(r \nu)-\Phi_r^{\rm TF}(r\nu)) \frac{\d \nu}{4\pi}.
\end{align*}
Therefore, \eqref{eq:int-rho-1} follows immediately from \eqref{eq:assume-D}. 

Next, using $\varphi^{\rm TF}(x) \le C |x|^{-4}$ and $\rho^{\rm TF}(x)\le C|x|^{-6}$ from Theorem \ref{thm:TF-Z}, we have
$$
\Phi_{|x|}^{\rm TF}(x)= \varphi^{\rm TF}(x) + \int_{|y|>|x|} \frac{\rho^{\rm TF}(y)}{|x-y|} \d y \le C|x|^{-4}, \quad \forall |x|>0.
$$
(The bound on the integral can be obtained from the pointwise bound on $\rho^{\rm TF}$, for instance, by Newton's theorem.) In particular, $\Phi_r^{\rm TF}(x)\le Cr^{-4}$ for all $|x|=r$. Therefore, by assumption \eqref{eq:assume-D}, 
$$ 
\Phi_{r}(x) \le  (\Phi_{r}(x) -\Phi_{r}^{\rm TF}(x)) + \Phi_{r}^{\rm TF}(x) \le C r^{-4}, \quad \forall |x|=r.
$$
Since $\Phi_{r}(x)$ is harmonic for $|x|>r$ and vanishing at infinity, we can apply Lemma \ref{lem:harmonic} and obtain  \eqref{eq:int-rho-2}:
$$ \sup_{|x| \ge r}|x|\Phi_{r}(x) = \sup_{|x| = r} |x| \Phi_{r}(x) \le Cr^{-3}.
$$

Now we turn to prove \eqref{eq:int-rho-3} and \eqref{eq:int-rho-4}. Let us consider the case $2r\le D$ first. From \eqref{eq:int-rho-1} and the bound $\rho^{\rm TF}\le C |x|^{-6}$ we have
\begin{align} \label{eq:int-rho-r-r/2}
\int_{2r>|x|>r/2} \rho_0 &=  \int_{|x|<2r} (\rho_0 - \rho^{\rm TF}) - \int_{|x|<r/2} (\rho_0 - \rho^{\rm TF}) + \int_{2r>|x|>r/2} \rho^{\rm TF} \nn\\
& \le \beta (2r)^{-3} +  \beta (r/2)^{-3}  + Cr^{-3} \le C r^{-3}.
\end{align}
Using Lemma \ref{lem:outside-kinetic} with $\lambda=1/2$, then using \eqref{eq:int-rho-2}  and  \eqref{eq:int-rho-r-r/2}, we deduce  that
\begin{align} \label{eq:int-rho-r-a}
\int \chi_{3r/2}^+(\rho_0)^{5/3} &\le C r^{-2} \int \chi_{r/2}^+ \rho_0  +  C r^3 \Big[ \sup_{|z| \ge r/2}  \Phi_{r/2}(z) \Big]_+^{5/2} \nn \\
&\qquad \qquad\qquad\qquad\qquad\qquad + C \Big[\sup_{|z| \ge r} |z| \Phi_{r}(z)\Big]_+^{7/3}\nn \\
&\le C r^{-2} \int \chi_r^+ \rho_0 + Cr^{-7}.
\end{align}
On the other hand, by using Lemma \ref{lem:binding} with $\lambda=1/2$ and $s=r^3$, then using \eqref{eq:int-rho-2} and \eqref{eq:int-rho-r-r/2}, we get 
\begin{align} \label{eq:int-rho-r-aa}
\int\chi_r^+ \rho_0  &\le C \int_{r \le |x| \le 2 r} \rho_0  +  C \Big[\sup_{|z| \ge r} |z| \Phi_r(z) \Big]_+  \nn\\
 & \qquad \qquad + C(s^{-1}+s)+ Cs^{6/5}   \| \chi_{3r/2}^+ \rho_0 \|_{L^{5/3}} \nn \\
& \le C r^{-3} + Cr^{18/5}   \| \chi_{3r/2}^+ \rho_0 \|_{L^{5/3}}.
\end{align}
Inserting \eqref{eq:int-rho-r-aa} into the left side of \eqref{eq:int-rho-r-a} we obtain
$$
\int \chi_{3r/2}^+(\rho_0)^{5/3} \le C r^{-7} + Cr^{8/5}   \| \chi_{3r/2}^+ \rho_0 \|_{L^{5/3}}
$$
which implies that
\bq \label{eq:int-rho-r-b}
\int \chi_{3r/2}^+(\rho_0)^{5/3} \le Cr^{-7}, \quad \forall r\in (0,D/2].
\eq
Then inserting \eqref{eq:int-rho-r-b} into \eqref{eq:int-rho-r-aa}, we find that
\bq \label{eq:int-rho-r-c}
\int \chi_{r}^+ \rho_0 \le Cr^{-3}, \quad \forall r\in (0,D/2].
\eq

Finally, let us conclude  \eqref{eq:int-rho-3}  and \eqref{eq:int-rho-4} for all $r\in (0,D]$. From \eqref{eq:int-rho-r-c} and the fact that $r\mapsto \chi_r^+$ is non-increasing, it follows that 
$$
\int \chi_{2a}^+ \rho_0 \le \int \chi_{a}^+ \rho_0 \le C a^{-3} = 8C (2a)^{-3}, \quad \forall  a\in (0,D/2].
$$
Changing the notation $r=2a$, we find that \eqref{eq:int-rho-4} holds true for all $r\in (0,D]$. Similarly, from \eqref{eq:int-rho-r-b} we have
$$
\int \chi_{2a}^+(\rho_0)^{5/3} \le \int \chi_{3a/2}^+(\rho_0)^{5/3}  \le Ca^{-7} = 2^7C (2a)^{-7}, \quad \forall a\in (0,D/2].
$$
Replacing $r=2a$, we find that \eqref{eq:int-rho-3} hold true for all $r\in (0,D]$.
\end{proof}

\noindent 
{\bf Step 2.} Let us introduce the exterior TF energy functional
$$
\cE_r^{\rm TF}(\rho)= c^{\rm TF}\int \rho^{5/3} - \int V_r \rho + \mathfrak{D}(\rho), \quad V_r(x)=\chi_r^+ \Phi_r(x).
$$

\begin{lemma} The TF  functional $\cE_r^{\rm TF}(\rho)$ has a unique minimizer $\rho_r^{\rm TF}$ over
$$
0\le \rho \in L^{5/3}(\R^3) \cap L^1(\R^3), \quad \int \rho \le \int \chi_r^+ \rho_0.
$$
This minimizer is supported on $\{|x|\ge r\}$ and satisfies the TF equation
$$
\frac{5c^{\rm TF}}{3} \rho_r^{\rm TF}(x)^{2/3} = [\varphi_r^{\rm TF}(x)-\mu_r^{\rm TF}]_+
$$
with $\varphi_r^{\rm TF}(x)= V_r(x)-\rho_r^{\rm TF}*|x|^{-1}$ and a constant $\mu_r^{\rm TF} \ge 0$. Moreover, if \eqref{eq:assume-D} holds true for some $\beta, D\in (0,1]$, then 
\bq \label{eq:rho-r-TF-5/3}
\int (\rho_r^{\rm TF})^{5/3} \le Cr^{-7}, \quad \forall r\in (0,D].
\eq
\end{lemma}

\begin{proof} The existence of $\rho_r^{\rm TF}$ and the TF equation follow from Theorem \ref{thm:TF-V} (i). From the TF equation and the fact that $\varphi_r^{\rm TF}(x) \le V_r(x) =0$ when $|x|<r$, we obtain 
$$\supp \rho_r^{\rm TF}\subset \{|x|\ge r\}.$$
Moreover, by the minimality of $\rho_r^{\rm TF}$ and \eqref{eq:int-rho-2} we have 
\begin{align*}
0 \ge \cE_r^{\rm TF}(\rho_r^{\rm TF}) &\ge c^{\rm TF}\int (\rho_r^{\rm TF})^{5/3} - Cr^{-3} \int \frac{\rho_r^{\rm TF}(x)}{|x|} \d x + \mathfrak{D}(\rho_r^{\rm TF}) \\
&\ge \frac{c^{\rm TF}}{2} \int (\rho_r^{\rm TF})^{5/3} - C (r^{-3})^{7/3}.
\end{align*}
Thus \eqref{eq:rho-r-TF-5/3} holds true. 
\end{proof}
\noindent 
{\bf Step 3.} Now we compare $\rho_r^{\rm TF}$ with $\chi_r^+\rho^{\rm TF}$. 

\begin{lemma} \label{lem:varphirTF-varphiTF} We can choose a universal constant $\beta>0$ small enough such that, if \eqref{eq:assume-D} holds true for some $D\in [Z^{-1/3},1]$, then $\mu_r^{\rm TF}=0$ and 
\begin{align*}
\left| \varphi_r^{\rm TF}(x) - \varphi^{\rm TF}(x) \right| \le C (r/|x|)^{\zeta}|x|^{-4}, \\
\left| \rho_r^{\rm TF}(x) - \rho^{\rm TF}(x) \right| \le C (r/|x|)^{\zeta}|x|^{-6}
\end{align*}
for all $r\in [Z^{-1/3},D]$ and for all $|x| \ge r$. Here $\zeta=(\sqrt{73}-7)/2\approx 0.77$.
\end{lemma}

\begin{proof} First, recall from Theorem \ref{thm:TF-Z} that when $|x|\ge r \ge Z^{-1/3}$, we have
\begin{align}
1  &\ge \frac{\varphi^{\rm TF}(x)}{A^{\rm TF} |x|^{-4}} \ge \left(1+ C \Big( \frac{r}{|x|}\Big)^\zeta\right)^{-2}, \label{eq:varphiTF-x>r}\\
1  & \ge \frac{\rho^{\rm TF}(x)}{\Big(\frac{3A^{\rm TF}}{5c^{\rm TF}}\Big)^{3/2} |x|^{-6}} \ge \left(1+ C \Big( \frac{r}{|x|}\Big)^\zeta\right)^{-3}.\label{eq:phiTF-x>r}
\end{align} 
In particular, from \eqref{eq:phiTF-x>r} we have $C|x|^{-6} \ge \rho^{\rm TF}(x) \ge C^{-1}|x|^{-6}$ for $|x|\ge r$, and hence
\bq \label{eq:phiTF-x>r-a}
Cr^{-3} \ge \int_{|x|\ge r} \rho^{\rm TF} \ge C^{-1} r^{-3}, \quad \forall r \ge Z^{-1/3}.
\eq

Now using Lemma \ref{lem:outside-energy-TF} with $\rho=\rho_r^{\rm TF}$ and the identity  
$$
\widetilde{\cE}_r(\rho)=\cE_r^{\rm TF}(\rho) + \int (\Phi_r-\Phi_r^{\rm TF})\rho,
$$
we find that 
\bq \label{eq:rTF-TF-1}
\cE_r^{\rm TF}(\chi_r^+\rho^{\rm TF}) \le \cE_r^{\rm TF}(\rho_r^{\rm TF}) - \int (\Phi_r-\Phi_r^{\rm TF})(\chi_r^+\rho^{\rm TF} - \rho_r^{\rm TF}).
\eq
Since $\Phi_r(x)-\Phi_r^{\rm TF}(x)$ is harmonic for $|x|>r$, we deduce from \eqref{eq:assume-D} that
$$\sup_{|x|\ge r} |\Phi_r(x)-\Phi_r^{\rm TF}(x)|= \sup_{|x|= r} |\Phi_r(x)-\Phi_r^{\rm TF}(x)| \le \beta r^{-4}.$$
Therefore,
$$
\left| \int (\Phi_r-\Phi_r^{\rm TF})(\chi_r^+\rho^{\rm TF} - \rho_r^{\rm TF}) \right| \le \sup_{|x|\ge r} |\Phi_r(x)-\Phi_r^{\rm TF}(x)| \int (\chi_r^+\rho^{\rm TF} + \rho_r^{\rm TF}) \le C\beta r^{-7}.
$$
Here we have used the upper bound in \eqref{eq:phiTF-x>r-a} and 
$$
 \int \rho_r^{\rm TF} \le \int_{|x|\ge r} \rho_0 \le Cr^{-3}, 
$$
which follows from the definition of $\rho_r^{\rm TF}$ and \eqref{eq:int-rho-4}. Hence, \eqref{eq:rTF-TF-1} reduces to 
\bq \label{eq:rTF-TF-1a}
\cE_r^{\rm TF}(\chi_r^+\rho^{\rm TF}) \le \cE_r^{\rm TF}(\rho_r^{\rm TF}) + C\beta r^{-7}.
\eq
We want to compare $\chi_r^+\rho^{\rm TF}$ with $\rho_r^{\rm TF}$ using the minimality property of the latter. In order to do so, we need to verify that the constraint $\int \rho \leq \int \chi_r^+ \rho_0$ is `almost' satisfied by $\chi_r^+\rho^{\rm TF}$. By `almost satisfied' we mean up to multiplication by a constant close to 1. Using $N\ge Z$, \eqref{eq:int-rho-1} and \eqref{eq:phiTF-x>r-a}, we have
$$
\int \chi_r^+ (\rho^{\rm TF}-\rho_0) = Z - N - \int_{|x|<r}  (\rho^{\rm TF}-\rho_0) \le \beta r^{-3} \le C\beta \int \chi_r^+ \rho^{\rm TF}.
$$
This can be rewritten as 
\bq \label{eq:rTF-TF-1a-mass}
\int (1-C\beta) \chi_r^+ \rho^{\rm TF} \le \int \chi_r^+ \rho_0. 
\eq
In the following, we choose $\beta>0$ small enough such that 
$$C\beta\le 1/2.$$ 
Since $\rho \mapsto \int \rho^{5/3} + \mathfrak{D}(\rho)$ is monotone, using \eqref{eq:int-rho-2} and \eqref{eq:phiTF-x>r-a} we can estimate 
$$
\cE_r^{\rm TF}\Big( (1-C\beta) \chi_r^+ \rho^{\rm TF} \Big) - \cE_r^{\rm TF}\Big(\chi_r^+ \rho^{\rm TF} \Big)  \le C\beta \int \Phi_r \chi_r^+ \rho^{\rm TF} \le C \beta r^{-7}.
$$  
Therefore, from \eqref{eq:rTF-TF-1a} we deduce that 
$$
\cE_r^{\rm TF}((1-C\beta) \chi_r^+ \rho^{\rm TF}) \le \cE_r^{\rm TF}(\rho_r^{\rm TF}) + C\beta r^{-7}.
$$
Combining with \eqref{eq:rTF-TF-1a-mass} and the minimality of $\rho_r^{\rm TF}$, we obtain
$$
\cE_r^{\rm TF}((1-C\beta) \chi_r^+ \rho^{\rm TF}) + \cE_r^{\rm TF}(\rho_r^{\rm TF})  - 2  \cE_r^{\rm TF} \Big( \frac{(1-C\beta) \chi_r^+ \rho^{\rm TF}+\rho_r^{\rm TF}}{2}\Big) \le C\beta r^{-7}.
$$
By the convexity of $\rho^{5/3}$ and $\mathfrak{D}(\rho)$ (see \eqref{eq:cv-D-TF}), we deduce that 
\bq \label{eq:cv-D}
\mathfrak{D}((1-C\beta)\chi_r^+\rho^{\rm TF}-\rho_r^{\rm TF}) \le C \beta r^{-7}.
\eq
For later purposes we also record that we can deduce that
\begin{align} \label{eq:cv-rho5/3}
\int \Big[ \big(  (1-C\beta) &  \chi_r^+(x)  \rho^{\rm TF}(x) \big)^{5/3} +(\rho_r^{\rm TF}(x))^{5/3} \nn \\
- 2 \Big( &\frac{(1-C\beta) \chi_r^+(x) \rho^{\rm TF}(x)+\rho_r^{\rm TF}(x)}{2}  \Big)^{5/3} \Big] \d x \le C\beta r^{-7}.
\end{align}
From  \eqref{eq:cv-D} and \eqref{eq:phiTF-x>r}, we find that
\begin{align} \label{eq:cv-D-1}
\mathfrak{D}(\chi_r^+\rho^{\rm TF}-\rho_r^{\rm TF}) &\le 2\mathfrak{D}(\chi_r^+\rho^{\rm TF} - (1-C\beta)\chi_r^+\rho^{\rm TF}) + 2 \mathfrak{D}((1-C\beta)\chi_r^+\rho^{\rm TF}-\rho_r^{\rm TF}) \nn\\
& \le (C\beta)^2 \mathfrak{D}(\chi_r^+\rho^{\rm TF}) +  C\beta r^{-7} \le C \beta r^{-7}, 
\end{align}
where the last inequality follows from choosing $\beta \le 1$.

Now we apply the Coulomb estimate \eqref{eq:Coulomb-estimate-1} with $f=\pm (\chi_r^+\rho^{\rm TF}-\rho_r^{\rm TF})$, then use  \eqref{eq:cv-D-1}, \eqref{eq:rho-r-TF-5/3} and $\int (\chi_r^+ \rho^{\rm TF})^{5/3} \le Cr^{-7}$ (by \eqref{eq:phiTF-x>r}). This leads to 
\begin{align*}
|(\chi_r^+\rho^{\rm TF}-\rho_r^{\rm TF})*|x|^{-1}| &\le C \| \chi_r^+\rho^{\rm TF}-\rho_r^{\rm TF}\|_{L^{5/3}}^{5/7} (\mathfrak{D}(\chi_r^+\rho^{\rm TF}-\rho_r^{\rm TF}))^{1/7} \\
& \le C \beta^{1/7} r^{-4}, \quad \forall |x|>0.
\end{align*}
Combining this with assumption \eqref{eq:assume-D}, we get 
\begin{align*}
\left| \varphi_r^{\rm TF}(x)-\varphi^{\rm TF}(x) \right| &=  \left| \Phi_r(x)- \Phi_r^{\rm TF}(x) +  (\chi_r^+\rho^{\rm TF}-\rho_r^{\rm TF})*|x|^{-1}  \right| \\
&\le C (\beta + \beta^{1/7}) r^{-4}, \quad \forall |x| \ge r.
\end{align*}
Note that $C r^{-4} \ge \varphi^{\rm TF}(x) \ge C^{-1} r^{-4}$ for $|x|\ge r$ by \eqref{eq:varphiTF-x>r}. Therefore, if $\beta>0$ is sufficiently small, we deduce that 
\bq \label{eq:varphi-r-TF-simple}
C r^{-4} \ge \varphi_r^{\rm TF}(x) \ge C^{-1} r^{-4}, \quad \forall |x| \ge r. 
\eq

In order to obtain a refined version of \eqref{eq:varphi-r-TF-simple}, we need to show that $\mu_r^{\rm TF}=0$. This will be done by using \eqref{eq:cv-rho5/3} and Theorem \ref{thm:TF-V}. Note that since $N\ge Z$,
$$
\lim_{|x|\to \infty} |x| \Phi_r(x) = Z - \int_{|y|\le r} \rho_0 \le \int \chi_r^+ \rho_0.
$$ 
Therefore, by Theorem \ref{thm:TF-V}, we can conclude that $\mu_r^{\rm TF}=0$ if  
\bq \label{eq:mu<varphi}
\mu_r^{\rm TF}< \inf_{|x|=r} \varphi_r^{\rm TF}(x).
\eq
Assume that \eqref{eq:mu<varphi} fails to hold. Then from \eqref{eq:varphi-r-TF-simple} we find that
$$\mu_r^{\rm TF}\ge \inf_{|x|=r} \varphi_r^{\rm TF}(x) \ge C^{-1} r^{-4}.$$
On the other hand, $\varphi_r^{\rm TF}(x)\le \Phi_r(x) \le Cr^{-3}|x|^{-1}$ by \eqref{eq:int-rho-2}. Therefore, from the TF equation 
$$
\frac{5}{3}\rho_r^{\rm TF}(x)^{2/3}=[\varphi_r^{\rm TF}(x)-\mu_r^{\rm TF}]_+ \le [C r^{-3}|x|^{-1} - C^{-1} r^{-4}]_+
$$
we find that $\rho_r^{\rm TF}(x)=0$ when $|x| \ge Cr$. Since the integrand in \eqref{eq:cv-rho5/3} is point-wise nonnegative, we can restrict the integral on $|x|\ge Cr$. Then using $\rho_r^{\rm TF}(x)=0$ when $|x| \ge Cr$, we deduce from \eqref{eq:cv-rho5/3} that
$$
\int_{|x|\ge Cr} \Big( (1-C\beta)\rho^{\rm TF} (x)\Big)^{5/3} \d x \le C \beta r^{-7}. 
$$
On the other hand, using $\rho^{\rm TF}(x)\ge C^{-1}|x|^{-6}$ (by \eqref{eq:phiTF-x>r}) we see that 
$$
\int_{|x|\ge Cr} \Big( (1-C\beta) \rho^{\rm TF}(x) \Big)^{5/3} \d x \ge  (1-C\beta)^{5/3}C^{-1} r^{-7}. 
$$
Putting the latter two estimates together, we obtain a contradiction if $\beta>0$ is sufficiently small. Thus in conclusion, we can choose $\beta>0$ small enough to ensure that $\mu_r^{\rm TF}=0$. 

Since $\mu_r^{\rm TF}=0$, we can apply the Sommerfeld estimate \eqref{eq:Sommerfeld} in Theorem \ref{thm:TF-V}. This allows us to improve \eqref{eq:varphi-r-TF-simple} to the sharp form
\begin{align}
1+ C \Big(\frac{r}{|x|} \Big)^\zeta & \ge \frac{\varphi_r^{\rm TF}(x)}{A^{\rm TF} |x|^{-4}} \ge  \left(1+ C \Big(\frac{r}{|x|} \Big)^\zeta \right)^{-2}, \quad \forall |x|\ge r.\label{eq:varphirTF-x>r}
\end{align} 
From \eqref{eq:varphirTF-x>r} and the TF equation, we have
\begin{align}
\Big(1+ C \Big(\frac{r}{|x|} \Big)^\zeta \Big)^{3/2} & \ge \frac{\rho_r^{\rm TF}(x)}{\Big(\frac{3 A^{\rm TF} }{5 c^{\rm TF}}\Big)^{3/2} |x|^{-6}} \ge \left(1+ C \Big(\frac{r}{|x|} \Big)^\zeta \right)^{-3}, \quad \forall |x|\ge r. \label{eq:phirTF-x>r}
\end{align} 
The desired estimates then follow by comparing \eqref{eq:varphiTF-x>r}-\eqref{eq:phiTF-x>r} with \eqref{eq:varphirTF-x>r}-\eqref{eq:phirTF-x>r}, respectively. 
\end{proof}

\noindent
{\bf Step 4.} In this step, we compare $\rho_r^{\rm TF}$ with $\chi_r^+\rho_0$.

\begin{lemma} \label{lem:DrhorTF-rho0}Let $\beta>0$ be as in Lemma \ref{lem:varphirTF-varphiTF}. Assume that \eqref{eq:assume-D} holds true for some $D\in [Z^{-1/3},1]$. Then 
$$
\mathfrak{D}(\rho_r^{\rm TF}-\chi_r^+\rho_0)\le C r^{-7+14/37}, \quad \forall r\in [Z^{-1/3},D].
$$
\end{lemma}

\begin{proof} We will use the notations in Lemma \ref{lem:outside-energy}. Since 
$$\int \eta_r^2 \rho_r^{\rm TF} \le \int \rho_r^{\rm TF} \le \int \chi_r^+ \rho_0,$$
we can choose $\rho=\eta_r^2 \rho_r^{\rm TF}$ as a trial state for $\cE^{\rm A}_r$ in Lemma \ref{lem:outside-energy}. This gives
\bq \label{eq:etar-rho0-rhor-cEA}
\cE^{\rm A}_r(\eta_r^2 \rho_0) \le  \cE^{\rm A}_r(\eta_r^2 \rho^{\rm TF}_r) + \mathcal{R}.
\eq
Note that from the estimate on $\mathcal{R}$ in Lemma \ref{lem:outside-energy} and \eqref{eq:int-rho-2}-\eqref{eq:int-rho-3}-\eqref{eq:int-rho-4} we get
\begin{align*}
\mathcal{R} &\le  C(1+(\lambda r)^{-2}) \int_{|x| \ge (1-\lambda)r} \rho_0 + C \lambda r^3 \Big[ \sup_{|z|=(1-\lambda)r}  \Phi_{|z|}(z) \Big]_+^{5/2} + \int (\eta_r^2 \rho_0)^{4/3} \\
&\le  C(\lambda^{-2}r^{-5} +   \lambda r^{-7})
\end{align*}
for all $\lambda \in (0,1/2]$. Note that $\eta_r$ also depends on $\lambda$. 

Next, let us bound $\cE^{\rm A}_r(\eta_r^2 \rho^{\rm TF}_r)$. By the definition of $\cE^{\rm A}_r$, we have
\bq \label{eq:cEA-cETF}
\cE^{\rm A}_r(\eta_r^2 \rho^{\rm TF}_r) \le \cE^{\rm TF}_r(\rho^{\rm TF}_r) + \int |\nabla (\eta_r^2 \rho^{\rm TF}_r)^{1/2}|^2 + \int \Phi_r (1-\eta_r^2) \rho^{\rm TF}_r.
\eq
Using $\Phi_r(x) \le Cr^{-4}$ from \eqref{eq:int-rho-2} and the bound $\rho_r^{\rm TF}(x)\le C|x|^{-6}$ from \eqref{eq:phirTF-x>r}, we have
$$
\int \Phi_r (1-\eta_r^2) \rho^{\rm TF}_r \le Cr^{-10} \int_{r\le |x|\le (1+\lambda) r} \le C \lambda r^{-7}.
$$
Now we consider the gradient term. We have
\begin{align*}
 \int |\nabla (\eta_r^2 \rho^{\rm TF}_r)^{1/2}|^2 \le 2 \int |\nabla \eta_r|^2 \rho^{\rm TF}_r + 2 \int \eta_r^2 |\nabla (\rho_r^{\rm TF})^{1/2}|^2.
\end{align*}
For the first term, using $|\nabla \eta_r|^2\le C (\lambda r)^{-2}\1(r\le |x|\le (1+\lambda)r)$ and the bound  $\rho_r^{\rm TF}(x)\le C|x|^{-6}$ from \eqref{eq:phirTF-x>r}, we have
$$
 \int |\nabla \eta_r|^2 \rho^{\rm TF}_r  \le C  (\lambda r)^{-2} r^{-6} \int_{r\le |x| \le (1+\lambda)r} \le C \lambda^{-1} r^{-5}.
$$
For the second term, by the TF equation and the bound  $\varphi_r^{\rm TF}(x) \ge C^{-1}|x|^{-4}$ in \eqref{eq:varphirTF-x>r}, 
\begin{align*}
\int \eta_r^2 |\nabla (\rho_r^{\rm TF})^{1/2}|^2 &= \left(\frac3{5c^{\rm TF}}\right)^{3/2} \int \eta_r^2 |\nabla (\varphi_r^{\rm TF})^{3/4}|^2 \\
& = \left(\frac3{5c^{\rm TF}}\right)^{3/2} \int \eta_r^2 |\varphi_r^{\rm TF}|^{-1/2} |\nabla \varphi_r^{\rm TF}|^2 \\
&\le C  \int |x|^2 \eta_r^2  |\nabla \varphi_r^{\rm TF}|^2.
\end{align*}
Then integrating by parts and using the TF equation again, we get
\begin{align*}
\int |x|^2 \eta_r^2 |\nabla \varphi_r^{\rm TF}|^2 &= \int |x|^2 \eta_r^2 (\nabla \varphi_r^{\rm TF}) \cdot (\nabla \varphi_r^{\rm TF})  = - \int (\nabla \cdot (|x|^2 \eta_r^2 \nabla \varphi_r^{\rm TF}))\varphi_r^{\rm TF}   \\
&= - \int (\nabla (|x|^2\eta_r^2)) \cdot (\nabla \varphi_r^{\rm TF}) \varphi_r^{\rm TF} - \int (|x|^2\eta_r^2) (\Delta \varphi_r^{\rm TF}) \varphi_r^{\rm TF}  \\
& = \frac{1}{2}\int (\Delta (|x|^2\eta_r^2)) (\varphi_r^{\rm TF})^2 - 4\pi \Big(\frac{3}{5c^{\rm TF}}\Big)^{3/2} \int |x|^2\eta_r^2 (\varphi_r^{\rm TF})^{5/2}\\
& \le C \lambda^{-1}r^{-5}.
\end{align*}
Here the last estimate follows from $0\le \varphi_r^{\rm TF}(x) \le C |x|^{-4}$ by \eqref{eq:varphirTF-x>r} and
\begin{align*}
|\Delta (|x|^2 \eta_r^2)| &= | 6\eta_r^2 + 4 x \cdot \nabla (\eta_r^2) + |x|^2 \Delta (\eta_r^2)| \\
&\le C \1(|x|\ge r) + C \lambda^{-2} \1((1+\lambda)r \ge |x| \ge r).
\end{align*}
Thus in summary, we can bound the gradient term as 
\begin{align*}
\int \eta_r^2 |\nabla (\rho^{\rm TF})^{1/2}|^2 \le C \int |x|^2 \eta_r^2 |\nabla \varphi_r^{\rm TF}|^2 \le C\lambda^{-1}r^{-5}.
\end{align*}
Therefore, \eqref{eq:cEA-cETF} reduces to
$$
\cE^{\rm A}_r(\eta_r^2 \rho^{\rm TF}_r) \le \cE^{\rm TF}_r(\rho^{\rm TF}_r) + C(\lambda^{-1}r^{-5}+\lambda r^{-7}).
$$
Combining with \eqref{eq:etar-rho0-rhor-cEA}, we deduce that
\begin{align*}
\cE_r^{\rm TF}(\eta_r^2 \rho_0) \le \cE^{\rm A}_r(\eta_r^2 \rho^{\rm TF}_r) + \mathcal{R} \le  \cE^{\rm TF}_r(\rho^{\rm TF}_r) + C(\lambda r^{-7}+\lambda^{-2}r^{-5}).
\end{align*}

By the minimality of $\rho_r^{\rm TF}$  and the convexity of $\cE_r^{\rm TF}(\rho)$, we can argue similarly to \eqref{eq:cv-D-TF} and conclude that 
$$ \mathfrak{D}(\eta_r^2 \rho_0 - \rho^{\rm TF}_r) \le C(\lambda r^{-7}+\lambda^{-2}r^{-5}).$$
On the other hand, by using the Hardy-Littewood-Soloblev inequality and \eqref{eq:int-rho-3}, we get 
\begin{align*}
\mathfrak{D}(\chi_r^+ \rho_0 - \eta_r^2 \rho_0) &\le \mathfrak{D}(\1\big((1+\lambda)r \ge |x| \ge r\big) \rho_0) \\
&\le C \| \1\big((1+\lambda)r \ge |x| \ge r\big) \rho_0\|_{L^{6/5}}^2\\
&\le C \left( \int \chi_r^+ \rho_0^{5/3}\right)^{6/5} \left( \int_{(1+\lambda)r \ge |x| \ge r} \right)^{7/15} \\
&\le C (r^{-7})^{6/5} (\lambda r^3)^{7/15} = C\lambda^{7/15} r^{-7}.
\end{align*}
Therefore, 
\begin{align*}
\mathfrak{D}(\chi_r^+\rho_0 - \rho_r^{\rm TF}) &\le 2 \mathfrak{D}(\chi_r^+\rho_0 - \eta_r^2 \rho_0) + 2\mathfrak{D}(\eta_r^2 \rho_0 - \rho_r^{\rm TF}) \\
&\le C(\lambda^{7/15} r^{-7}+\lambda^{-2}r^{-5})
\end{align*}
for all $\lambda \in (0,1/2]$. We can choose $\lambda \sim r^{30/37}$ and conclude that
\begin{align*}
\mathfrak{D}(\chi_r^+\rho_0 - \rho_r^{\rm TF}) \le C r^{-7+14/37 }.
\end{align*}
This is the desired estimate.
\end{proof}

\noindent
{\bf Step 5.} Now we are ready to conclude.  

\begin{proof}[Proof of Lemma \ref{thm:screened-it}] Let $\beta>0$ be as in Lemma \ref{lem:varphirTF-varphiTF} and assume that \eqref{eq:assume-D} holds true for some $D\in [Z^{-1/3},1]$.

Let $r\in [Z^{-1/3},D]$ and $|x|\ge r$. Recall the decomposition \eqref{eq:rev-decom}:
\begin{align*}
\Phi_{|x|}(x)-\Phi_{|x|}^{\rm TF}(x) &= \varphi_r^{\rm TF}(x) -\varphi^{\rm TF}(x)+ \int_{|y|>|x|} \frac{\rho_r^{\rm TF}(y)-\rho^{\rm TF}(y)}{|x-y|} \d y,\\
&\qquad + \int_{|y|<|x|}  \frac{\rho_r^{\rm TF}(y)-(\chi_r^+\rho_0)(y)}{|x-y|} \d y.
\end{align*}
By Lemma \ref{lem:varphirTF-varphiTF}, we have
\begin{align*}
\left| \varphi_r^{\rm TF}(x) - \varphi^{\rm TF}(x) \right| \le C (r/|x|)^{\zeta}|x|^{-4}
\end{align*}
and
\begin{align*}
\int_{|y|>|x|} \frac{|\rho_r^{\rm TF}(y)-\rho^{\rm TF}(y)|}{|x-y|} \d y \le C \int_{|y|>|x|} \frac{(r/|y|)^{\zeta}|y|^{-6}}{|x-y|} \d y \le C (r/|x|)^{\zeta}|x|^{-4}.
\end{align*}
Moreover, from \eqref{eq:Coulomb-estimate-2}, \eqref{eq:int-rho-3}, \eqref{eq:rho-r-TF-5/3} and Lemma \ref{lem:DrhorTF-rho0}, we get
\begin{align*}
\left| \int_{|y|<|x|}  \frac{\rho_r^{\rm TF}(y)-(\chi_r^+\rho_0)(y)}{|x-y|} \d y \right| &\le C \| \rho_r^{\rm TF}-\chi_r^+\rho_0\|_{L^{5/3}}^{5/6} \Big(|x| \mathfrak{D}(\rho_r^{\rm TF}-\chi_r^+\rho_0)\Big)^{1/12} \\
&\le C (r^{-7})^{1/2} (|x| r^{-7+14/37} )^{1/12} \\
& =   C (|x|/r)^{4+23/444} |x|^{-4+1/12-23/444}   ,
\end{align*}
Thus in summary, for all $r\in [Z^{-1/3},D]$ and $|x|\ge r$, we have
\bq \label{eq:bootstrap-total}|\Phi_{|x|}(x)-\Phi_{|x|}^{\rm TF}(x)|\le C (r/|x|)^\zeta |x|^{-4}+ C (|x|/r)^{5} |x|^{-4+1/12-23/444}.\eq

Now let us conclude using \eqref{eq:bootstrap-total}. We fix a universal constant $\delta \in (0,1)$ such that
$$
\frac{1}{12} -\frac{23}{444} - \frac{10\delta}{1-\delta}>0
$$
(the reason for that will be clear later). We distinguish two cases. 

{\bf Case 1:} $D^{1+\delta}\le Z^{-1/3}$. In this case, we simply use the initial step. Indeed, for all
$$|x|\le D^{1-\delta}\le (Z^{-1/3})^{(1-\delta)/(1+\delta)},$$
by Lemma \ref{thm:screened-first} we have
\bq \label{eq:bootstrap-total-1a}
|\Phi_{|x|}(x)-\Phi_{|x|}^{\rm TF}(x)|\le C_1 Z^{4/3} |x|^{1/12} \le C_1 |x|^{-4 +1/12-  8\delta/(1-\delta)}.
\eq

{\bf Case 2:} $D^{1+\delta}\ge Z^{-1/3}$. In this case, we use \eqref{eq:bootstrap-total} with $r=D^{1+\delta}$. For all $D \le |x| \le D^{1-\delta}$ we have
$$ |x|^{2\delta/(1-\delta)} \le r/|x| \le |x|^\delta.$$
Therefore, \eqref{eq:bootstrap-total} implies that 
\bq \label{eq:bootstrap-total-1b} |\Phi_{|x|}(x)-\Phi_{|x|}^{\rm TF}(x)|\le C |x|^{-4+\zeta \delta} + C|x|^{-4+1/36-10 \delta/(1-\delta)}.
\eq

From \eqref{eq:bootstrap-total-1a} and \eqref{eq:bootstrap-total-1b}, we conclude that in both cases,
$$ |\Phi_{|x|}(x)-\Phi_{|x|}^{\rm TF}(x)|\le C |x|^{-4+\eps}, \quad \forall D \le |x| \le D^{1-\delta}.$$
with  
$$\eps:=\min\Big\{\zeta\delta, \frac{1}{12} - \frac{23}{444} - \frac{10\delta}{1-\delta} \Big\}>0.$$This completes the proof of Lemma \ref{thm:screened-it}.
\end{proof}

\section{Proof of the main results} \label{sec:proof-main-result}

\begin{proof}[Proof  of Theorem \ref{thm:ionization}] Since we have proved $N\le 2Z + CZ^{2/3}+C$ in Lemma \ref{lem:2Z}, it remains to consider the case $N\ge Z\ge 1$. By Lemma \ref{lem:screened}, we can find universal constants $C, \eps, D>0$ such that 
$$
|\Phi_{|x|}(x)  - \Phi^{\rm TF}_{|x|}(x)  | \le C |x|^{-4+\eps}, \quad \forall |x|\le D.
$$
In particular, \eqref{eq:assume-D} holds true with a universal constant $\beta=C D^\eps$. We can choose $D$ sufficiently small such that $D\leq 1$ and $\beta\leq 1$, which allow us to apply Lemma \ref{lem:screened-easy-bounds}. Then using \eqref{eq:int-rho-1} and \eqref{eq:int-rho-4} with $r=D$, we find that 
$$
\int_{|x|>D} \rho_0 + \left| \int_{|x|<D} (\rho_0 - \rho^{\rm TF}) \right|   \le C.
$$
Combining  with $\int\rho^{\rm TF}=Z$, we obtain the ionization bound
\begin{equation*}
N = \int \rho_0 = \int_{|x|>D} \rho_0 + \int_{|x|<D} (\rho_0-\rho^{\rm TF}) + \int_{|x|<D} \rho^{\rm TF} \le C + Z. \qedhere
\end{equation*}
\end{proof}

\begin{proof}[Proof of Theorem \ref{thm:screened-intro}] By Lemma \ref{lem:screened}, we can find universal constants $C, \eps, D>0$ such that 
$$
|\Phi_{|x|}(x)  - \Phi^{\rm TF}_{|x|}(x)  | \le C |x|^{-4+\eps}, \quad \forall |x|\le D.
$$
As in the previous proof we can assume $D\leq 1$ and $CD^{\eps}\leq 1$ in order to apply Lemma \ref{lem:screened-easy-bounds}.

It remains to consider the case when $|x|>D$. We decompose
\begin{align} \label{eq:proof-screened-intro-b}
 \Phi_{|x|}(x) - \Phi_{|x|}^{\rm TF}(x) =  \Phi_{D}(x) - \Phi_{D}^{\rm TF}(x) & + \int_{|x|>|y|>D} \frac{\rho^{\rm TF}(y)-\rho_0(y)}{|x-y|} \d y.
 \end{align}
Since $\Phi_{D}(x) - \Phi_{D}^{\rm TF}(x)$ is harmonic for $|x|>D$ and vanishing at infinity, by Lemma \ref{lem:harmonic} we find that
$$
\sup_{|x|\ge D} |\Phi_{D}(x) - \Phi_{D}^{\rm TF}(x)| \le \sup_{|x|=D} |\Phi_{D}(x) - \Phi_{D}^{\rm TF}(x)| \le C D^{-4+\eps}.
$$
Moreover, using $\rho^{\rm TF}(y) \le C |y|^{-6}$ (by Theorem \ref{thm:TF-Z}), we can estimate
$$
\int_{|x|>|y|>D} \frac{\rho^{\rm TF}(y)}{|x-y|} \d y \le C \int_{|x|>|y|>D} \frac{|y|^{-6}}{|x-y|} \d y \le CD^{-4}.
$$
Finally, using \eqref{eq:int-rho-3} and \eqref{eq:int-rho-4} in Lemma \ref{lem:screened-easy-bounds}, we have
\begin{align*}
&\int_{|x|>|y|>D} \frac{\rho_0(y)}{|x-y|} \d y \le \int_{|y|>D, |x-y|>D}   \frac{\rho_0(y)}{|x-y|} \d y + \int_{|y|>D\ge |x-y|}   \frac{\rho_0(y)}{|x-y|} \d y  \\
& \le \int_{|y|>D}   \frac{\rho_0(y)}{D} \d y + \left(\int_{|y|>D} \rho_0(y)^{5/3} \d y \right)^{3/5} \left( \int_{D\ge |x-y|} \frac{\d y}{|x-y|^{5/2}}  \right)^{2/5} \\
&\le CD^{-4} + C(D^{-7})^{3/5} (D^{1/2})^{2/5} \le CD^{-4}.
\end{align*}
Thus from \eqref{eq:proof-screened-intro-b} we conclude that 
\begin{align*} 
 |\Phi_{|x|}(x) - \Phi_{|x|}^{\rm TF}(x)| \le CD^{-4} , \quad \forall |x| > D.
\end{align*}
In summary, 
\begin{align*}
|\Phi_{|x|}(x) - \Phi_{|x|}^{\rm TF}(x)| \le C|x|^{-4+\eps} + CD^{-4}, \quad \forall |x|>0.
 \end{align*}
This is the desired estimate. 
\end{proof}

\begin{proof}[Proof of Theorem \ref{thm:radius}] 
By Lemma \ref{lem:screened}, we can find universal constants $C, \eps, D>0$ such that 
$$
|\Phi_{|x|}(x)  - \Phi^{\rm TF}_{|x|}(x)  | \le C |x|^{-4+\eps}, \quad \forall |x|\le D.
$$
We can assume $\eps\leq\zeta$, $D\leq 1$ and $CD^{\eps}\leq 1$. (Here $\zeta=(\sqrt{73}-7)/2 \approx 0.77$ from Theorem \ref{thm:TF-Z}.) From (\ref{eq:int-rho-4}) in Lemma \ref{lem:screened-easy-bounds} and 
the uniform bound $N\le Z+C$ in Theorem \ref{thm:ionization}, we get for all $r\in (0,D]$,
\begin{align*}
\left| \int_{|y|\ge r} (\rho_0(y)-\rho^{\rm TF}(y)) \d y \right| = \left| N- Z -  \int_{|y|< r} (\rho_0(y)-\rho^{\rm TF}(y)) \d y \right| \le Cr^{-3+\eps}.
\end{align*}
On the other hand, from the Sommerfeld estimate in Theorem \ref{thm:TF-Z} we have 
$$ \left| \rho^{\rm TF}(x) - \Big(\frac{3A^{\rm TF}}{5c^{\rm TF}}\Big)^{3/2} |x|^{-6}  \right| \le C |x|^{-6}  \Big(\frac{Z^{-1/3}}{|x|}\Big)^{\zeta}, \quad \forall |x|\ge Z^{-1/3}.$$
Integrating the latter estimate over $|x| > r \ge Z^{-1/6}$ and using 
$$
\Big(\frac{Z^{-1/3}}{|x|}\Big)^{\zeta} \le \Big(\frac{r^2}{r}\Big)^{\zeta} = r^\zeta \le r^\eps 
$$
we obtain 
$$
\left| \int_{|x|>r} \rho^{\rm TF}(x) \d x  - (B^{\rm TF}/r)^{3} \right| \le  Cr^{-3+\eps}, \quad \forall r\in [Z^{-1/6},D]
$$
where
$$
B^{\rm TF}= \Big(\frac{4\pi}{3}  \Big(\frac{3A^{\rm TF}}{5c^{\rm TF}}\Big)^{3/2} \Big)^{1/3} = 5 c^{\rm TF} \left(\frac{4}{3 \pi^2}\right)^{1/3}.
$$
Consequently,
\bq \label{eq:rBZD}
\left| \int_{|x|>r} \rho_0(x) \d x  - (B^{\rm TF}/ r)^{3} \right| \le  Cr^{-3+\eps},  \quad \forall r\in [Z^{-1/6},D].
\eq

Applying \eqref{eq:rBZD} with $r=D$ and $r=Z^{-1/6}$, we deduce that  
$$
\int_{|x|>D} \rho_0(x) \d x \le CD^{-3}, \quad \int_{|x|>Z^{-1/6}} \rho_0(x) \d x \ge C^{-1} Z^{1/2}. 
$$
Thus, if we restrict to the case $C^{-1} Z^{1/2}>\kappa>CD^{-3}$,  
$ R_\kappa:= R(N,Z,\kappa) \in [Z^{-1/6},D]$ and we can apply \eqref{eq:rBZD} with $r=R_\kappa$.
We obtain
$$
\left| \kappa  - (B^{\rm TF}/ R_\kappa)^{3} \right| \le  CR_\kappa^{-3+\eps}.
$$
We now have to transform this inequality into the desired form.
To this end, it can be rewritten as
$$
|t^3-1| \le C (t\kappa^{-1/3})^{\eps}
$$
with $t=\kappa^{1/3} R_\kappa/B^{\rm TF}$. Then using 
$$
|t-1| =\frac{|t^3-1|}{t^{2}+t+1} \le  \frac{|t^3-1|}{t^{\eps}}
$$
we conclude that 
$$
|\kappa^{1/3} R_\kappa/B^{\rm TF}-1| \le  C\kappa^{-\eps/3}.
$$
Thus in summary, if $\kappa>CD^{-3}$, then 
$$
\limsup_{N\ge Z \to \infty} |\kappa^{1/3} R_\kappa/B^{\rm TF}-1| \le C\kappa^{-\eps/3}.
$$
This is equivalent to the desired estimate. 
\end{proof}

\bibliographystyle{plain}

\end{document}